\documentclass[a4paper,12pt]{article}
\usepackage{amssymb}
\usepackage{amsfonts,amsmath,amstext,amsbsy,amsopn,amsthm,amscd}
\usepackage{dsfont}

\title{The Birman--Solomyak theorem revisited:  a novel elementary proof,
generalisation, and applications}
\author{V. Bach, A.F.M.~ter Elst and J. Rehberg}

\newcommand{\HS}{{\rm HS}}
\newcommand{\cb}{{\mathcal B}}
\newcommand{\cd}{{\mathcal D}}
\newcommand{\ch}{{\mathcal H}}
\newcommand{\cl}{\mathcal{L}}
\newcommand{\cm}{\mathcal{M}}
\newcommand{\cn}{\mathcal{N}}
\newcommand{\cs}{\mathcal{S}}

\newcommand{\Ni}{\mathds{N}}

\newcommand{\Ri}{\mathds{R}}
\newcommand{\R}{\mathds{R}}
\newcommand{\Ci}{\mathds{C}}

\newcommand{\one}{\mathds{1}}

\newcommand{\Tr}{{\mathop{\rm tr \,}}}

\newcommand{\gotl}{\mathfrak l}
\newcommand{\gotm}{\mathfrak m}
\newcommand{\gott}{\mathfrak t}

\newcounter{teller}

\newenvironment{tabel}{\begin{list}%
{\rm  (\alph{teller})\hfill}{\usecounter{teller} \leftmargin=1.1cm
\labelwidth=1.1cm \labelsep=0cm \parsep=0cm}
                      }{\end{list}}

\newtheorem{theorem}{Theorem}[section]
\newtheorem{lemma}[theorem]{Lemma}

\newtheorem{cor}[theorem]{Corollary}
\newtheorem{proposition}[theorem]{Proposition}
\newtheorem{prop}[theorem]{Proposition}

\theoremstyle{definition}

\numberwithin{equation}{section}

\DeclareMathOperator{\tr}{tr}     
   
\DeclareMathOperator{\supp}{supp}  
\DeclareMathOperator{\dive}{div}

\newcommand{\ca}{{\mathcal A}}
\newcommand{\spann}{\mathop{\rm span}}
\newcommand{\Q}{Q}
\newcommand{\M}{\mathcal M}

\newcommand{\ce}{{\mathcal E}}

\makeatletter
\def\eqnarray{\stepcounter{equation}\let\@currentlabel=\theequation
\global\@eqnswtrue
\tabskip\@centering\let\\=\@eqncr
$$\halign to \displaywidth\bgroup\hfil\global\@eqcnt\z@
$\displaystyle\tabskip\z@{##}$&\global\@eqcnt\@ne
  \hfil$\displaystyle{{}##{}}$\hfil
  &\global\@eqcnt\tw@ $\displaystyle{##}$\hfil
  \tabskip\@centering&\llap{##}\tabskip\z@\cr}

\def\endeqnarray{\@@eqncr\egroup
      \global\advance\c@equation\m@ne$$\global\@ignoretrue}

\def\@yeqncr{\@ifnextchar [{\@xeqncr}{\@xeqncr[5pt]}}
\makeatother

\begin{document}

\date{}

\maketitle

\begin{abstract}
\noindent
We provide a new short proof for the Birman--Solomyak theorem for Hilbert--Schmidt operators
and give an application to a Schr\"odinger--Poisson system.
\end{abstract} 

\vspace*{2em} 


\medskip 


\section{Introduction} \label{s-introd}

About sixty years ago, Birman and
Solomyak developed their calculus of double operator spectral
integrals, see \cite{BirmanSolomyak4} \cite{BirmanSolomyak2} \cite{BirmanSolomyak3} and derived, among
other results, estimates of the form
\begin{equation} \label{e-BirmSol1}
\|f(A) -f(B) \|_{\cs_p} 
\le 
L(f) \, \|A - B\|_{\cs_p},
\end{equation}
where $A, B$ are bounded, self-adjoint operators, $f\colon  \Ri \to \Ri$ is
a function from an adequate class, $L(f) < \infty$ is a suitable
constant depending only on $f$, and $\|\cdot \|_{\cs_p}$ is a Schatten
norm.
A directly accessible case of \eqref{e-BirmSol1} is $p=2$, where
the Lipschitz property of $f$ is sufficient for
\eqref{e-BirmSol1}, with $L(f)$ being the Lipschitz constant of~$f$.
Birman and Solomyak remarked in \cite{BirmanSolomyak} page~156 that, in
spite of this simple characterisation, an elementary proof of
\eqref{e-BirmSol1} for $p=2$ was lacking for a long time, and it was
eventually given by Gesztesy, Pushnitski and Simon more than forty
years later in \cite{GesztesyPushnitskiSimon}.
Their approach is based on a classical
theorem of L\"owner \cite{Lowner} for finite-dimensional operators
complemented by suitable limit arguments.
We became interested in
\eqref{e-BirmSol1} when one of us applied it in \cite{KaiserRehberg2}
\cite{KaiserRehberg} for the derivation of estimates for particle density
operators in the Kohn--Sham system.

The first aim of this paper is to give an elementary proof of
\eqref{e-BirmSol1} for $p=2$, which is quite different from the one in
\cite{GesztesyPushnitskiSimon} and to generalise \eqref{e-BirmSol1} for the case $p=2$
to arbitrary \emph{unbounded} self-adjoint operators, provided 
$f \colon \Ri \to \Ri$ is (globally) Lipschitz continuous.
More precisely, we prove
the following theorem, in which, for a given Hilbert space $\mathcal
H$, we denote by $\cl_\HS(\mathcal H)$ the Hilbert space of all
Hilbert--Schmidt operators and Hilbert--Schmidt norm $\|\cdot\|_\HS$.

\begin{theorem} \label{tBS101}
Let $A,B$ be two self-adjoint operators in a Hilbert space $H$.
Let $C \in \cl_\HS(H)$ and suppose that $A = B + C$.
Let $f \colon \Ri \to \Ri$ be a Lipschitz continuous function with 
Lipschitz constant $L$.
Then $D(A) = D(B) \subset D(f(A)) = D(f(B))$ and the operator 
$f(A) - f(B)$ extends to a Hilbert--Schmidt operator,
by abuse of notation still denoted by $f(A) - f(B)$, such that 
\begin{equation} \label{e-BirmSol1-thm-1.1}
\|f(A) - f(B)\|_\HS
\leq L \, \|C\|_\HS
. 
\end{equation} 
\end{theorem}

In our proof of Theorem~\ref{tBS101} we first assume that $f$ is
bounded and that $A$ and $B$ have pure point spectrum, 
which includes all compact operators and
those with compact resolvent.
In contrast to \cite{GesztesyPushnitskiSimon} our proof only 
requires Parseval's identity.
Having this
at hand, we employ an argument quite similar to one in \cite{GesztesyPushnitskiSimon}
to carry this over to the class of all self-adjoint operators $A,B$, 
still with the Hilbert--Schmidt condition.
Finally, an approximation
argument allows us to eliminate the boundedness assumption for $f$.
While our proof of \eqref{e-BirmSol1} in case that $p = 2$ does not
generalise to other Schatten classes with $p \neq 2$, 
we note that Peller \cite{Peller} succeeded in proving \eqref{e-BirmSol1} for 
$p \neq 2$ by using double and triple sums instead of the calculus of
double spectral integrals.

Secondly, we deduce from \eqref{e-BirmSol1-thm-1.1} estimates of the
form
\begin{equation} \label{e-Geszte}
\|f(A)  - f(B) \|_\HS 
\leq    
K \, \|(A +\lambda )^{-1} - (B +\lambda )^{-1} \|_\HS  ,
\end{equation}
where $K > 0$ depends on the Lipschitz constant of $f$ and
$\lambda$.
For smooth functions $f$ with a certain decay
the estimate~\eqref{e-Geszte} is also in
\cite{GesztesyNichols} (3.5) for general Schatten classes.
The proof in \cite{GesztesyNichols}
involves the Helffer--Sj\"ostrand calculus, which is much more complicated.
The importance of
\eqref{e-Geszte} lies in the fact that, e.g., for Schr\"odinger
operators, the difference of $A = -\dive \bigl ( \mathfrak m ^{-1} \nabla \bigr ) + U$ and
 $B = -\dive \bigl ( \mathfrak m ^{-1} \nabla \bigr )  + V$
is usually not a Hilbert--Schmidt operator, but the difference of
their resolvents often is.
Note that for the validity of
\eqref{e-Geszte} only Lipschitz continuity of $f$ is required.
In particular, $f$ does not need to be smooth, nor
compactly supported and no almost analytic extension of $f$
is required.
We further remark that an estimate similar to \eqref{e-Geszte} was
derived and used in \cite{KaiserRehberg2} and
\cite{KaiserRehberg} for the estimation of density matrices with
respect to Schr\"odinger potentials.

The third part of the paper is devoted to an application of the
Birman--Solomyak theorem in density functional
theory, namely, in the investigation of the Schr\"odinger--Poisson
system.
The connection to the
Birman--Solomyak theorem is the following.
Let $\Omega \subset \Ri^d$, with $d \in \{1,2,3\}$, be a bounded Lipschitz domain and 
$f \colon \Ri \to (0,\infty)$ a strictly decreasing differentiable 
function with suitable asymptotics at $\infty$.
Let $H = -\dive \bigl ( \mathfrak m ^{-1} \nabla \bigr ) $ be the unperturbed Schr\"odinger operator
with mixed boundary conditions on $\partial \Omega$.
If $U \in L_2(\Omega,\Ri)$,  then the spectrum of the Schr{\"o}dinger operator
$H + U$ on $L_2(\Omega)$ is purely discrete with sufficiently
rapidly growing eigenvalues, hence for all $\mu \in \Ri$ the operator
\[
f(H + U -\mu)
\] 
defines a
density matrix, that is a positive trace-class operator, and the
diagonal of its Schwartz kernel defines a particle density
$\mathcal{N}(U)\in L_2(\Omega,\Ri)$.
More precisely, viewing a
bounded function $W$ as a multiplication operator, one can
show that the map
\[
W \mapsto 
\Tr\big( f(H + U -\mu) \,  W \big)
\]
extends to a continuous linear functional on $L_2(\Omega)$, which can be
represented by a function $\mathcal{N}(U) \in L_2(\Omega)$.
 In physics, the number $\mu \in \mathbb{R}$ is a
Lagrange multiplier called the Fermi level.
It ensures the
normalisation condition $\Tr (f(H+ U -\mu)) = N$, where $N > 0$
is the particle number.
In many applications, $U = V_0 + V$ is a small
perturbation of a fixed reference potential $V_0 \in
L_2(\Omega)$.
Since the Schr\"odinger--Poisson system can be written
as one single equation (see \cite{Nier})
\begin{equation} \label{e-SP}
- \nabla \cdot \varepsilon \nabla V - \mathcal N(V_0 + V) 
= q  ,
\end{equation}
one is interested in the functional analytic properties of the 
map $U \mapsto \cn(U)$
from $L_2(\Omega)$ into $L_2(\Omega)$ and 
from $W^{1,2}(\Omega)$ into $W^{-1,2}(\Omega)$,
see also \cite{WoodsPayneHasnip} Subsection~3.5.
In \cite{Nier} it was proved that
the operator from $W^{1,2}(\Omega)$ into $W^{-1,2}(\Omega)$ 
is antimonotone and locally Lipschitz continuous, so
that the standard monotone theory with the Browder-Minty theorem
applies.
Unfortunately, the proofs in \cite{Nier} are technically
involved.
We will use the Birman--Solomyak theorem to prove that 
$U \mapsto \cn(U)$ is locally Lipschitz continuous 
from $L_2(\Omega)$ into $L_2(\Omega)$.
In \cite{KaiserNeidhardtRehberg}, however, the monotonicity of $U \mapsto \cn(U)$ 
from $L_2(\Omega)$ into $L_2(\Omega)$ and 
from $W^{1,2}(\Omega)$ into $W^{-1,2}(\Omega)$
was established by an
abstract operator-theoretic principle, without using specific
properties of the concrete Schr\"odinger operators which are involved.
Moreover,
the proof in \cite{KaiserNeidhardtRehberg} is fairly simple.

Concerning the continuity properties, in \cite{KaiserRehberg2} and
\cite{KaiserRehberg} local Lipschitz continuity of the map
$U \mapsto \cn(U)$ 
from $L_2(\Omega)$ into $L_2(\Omega)$
was established, which implies, trivially, the
local Lipschitz continuity from $W^{1,2}(\Omega)$ into $W^{-1,2}(\Omega)$.
These two
properties, (anti)monotonicity and local Lipschitz continuity, are
the key ingredients for an adequate investigation of the
Schr{\"o}dinger--Poisson system \eqref{e-SP}.
Since the latter is a
fundamental model in modern physics (see \cite{ParedesOlivieriMichinel} for a recent
overview), we recall the essentials of the analysis for this
system.
This includes both a derivation of the local Lipschitz
continuity for the particle density operator and a description of the
foundation of a numerical procedure, which is standard nowadays, and
rests on a functional-analytic based contraction principle.

As already indicated, several ideas of this paper are already present in the quoted
foregoing papers, frequently with complicated proofs.
It is, however, our intention to present an entirely 
self-contained text on the Schr\"odinger--Poisson system, now 
based on elementary proofs of the crucial properties of the particle density 
operator and readable also for theoretical  physicists.
Essentially new in this paper is a deeper
analysis of the Schr\"odinger operators. 
In particular it is proved that $H +V$ makes
sense also as operator sum, and not only via forms.
Moreover, a deeper analysis 
shows here that the resolvents of the operators $(H+U)$ map 
$L_2(\Omega)$ into $L_\infty(\Omega)$, 
and the corresponding norms are bounded, uniformly on $L_2$-bounded sets 
of potentials $U$.
Moreover, in Theorem~\ref{tBS510} we 
prove that the solution of the Schr\"odinger--Poisson system
as function of the reference potential $V_0$ and the right hand side $q$ also is 
locally Lipschitz. 
This is the key tool for the treatment of the Kohn--Sham system,
see  \cite{KaiserRehberg2} \cite{KaiserRehberg}.
We refer to \cite{WoodsPayneHasnip} Section~3.5 for the numerics of the 
Kohn--Sham system.

The paper is organised as follows.
In Section~\ref{s-proof} we prove Theorem~\ref{tBS101}.
In Section~\ref{s-resolvent} we prove (\ref{e-Geszte}) on resolvents.
We provide in Section~\ref{ss-scheme} an abstract theorem to construct a solution for a 
strongly monotone operator via the Banach contraction theorem.
In Theorem~\ref{tBS510} we give a rigorous statement and proof 
of the Schr\"odinger--Poisson system (\ref{e-SP}) and show that the solution is 
(locally) Lipschitz as function of the data.

\section{The elementary proof} \label{s-proof}

If $(X,\ca,\mu)$ is a measure space, then we say that $(X,\ca,\mu)$ is
a {\bf locally finite} if for all $E \in \ca$ with $\mu(E) = \infty$ 
there exists a $F \in \ca$ such that $F \subset E$ and $0 < \mu(F) < \infty$.
If $A$ is an operator in a Hilbert space $\mathcal H$, then we provide $D(A)$ with the graph norm given by 
$\|u\|_{D(A)}^2= \|A u\|_\mathcal H^2 + \|u\|_\mathcal H^2$.

We need some lemmas.

\begin{lemma} \label{lBS201}
Let $A$ be a self-adjoint operator in a Hilbert space $\mathcal H$ and $f \colon \Ri \to \Ri$ a Lipschitz continuous function.
Then one has the following.
\begin{tabel} 
\item \label{lBS201-1}
$D(A) \subset D(f(A))$.
\item \label{lBS201-2}
$D(A)$ is a core for  $f(A)$.
\item \label{lBS201-3}
For all $n \in \Ni$ define $\chi_n \colon \Ri \to \Ri$ by 
$\chi_n = (-n) \vee x \wedge n$.
Then $\lim_{n \to \infty} (\chi_n \circ f)(A) u = f(A) u$
for all $u \in D(f(A))$.  
\end{tabel}
\end{lemma}
\begin{proof}
Using the spectral theorem we may assume that $\mathcal H = L_2(X,\ca,\mu)$
and $A$ is the multiplication operator with a measurable function
$h \colon X \to \Ri$, where $(X,\ca,\mu)$ is a locally finite measure space.
Since $f$ is Lipschitz continuous, there are $L,b > 0$ such that 
$|f(x)| \leq L \, |x| + b$ for all $x \in \Ri$.
Now $f(A)$ is the multiplication operator with the function $f \circ h$.

`\ref{lBS201-1}'.
Let $u \in D(A)$.
Then $\int |h \, u|^2 < \infty$ and $\int |u|^2 < \infty$.
Hence 
\[
\int |(f \circ h) \cdot u|^2
\leq \int |L \, |h| + b|^2 \, |u|^2
\leq 2 L^2 \int |h \, u|^2 + 2 b^2 \int |u|^2 
< \infty
.  \]
Therefore $u \in D(f(A))$.

`\ref{lBS201-2}'.
Let $u \in D(f(A))$.
If $n \in \Ni$, then $u \, \one_{[|h| \leq n]} \in D(A)$.
Moreover, 
\[
\lim_{n \to \infty} \int |(f \circ h) \, u \, \one_{[|h| \leq n]}|^2
= \int |(f \circ h) \, u|^2
\]
by the monotone convergence theorem.
So 
\[
\lim_{n \to \infty} \|A( u - u \, \one_{[|h| \leq n]})\|_\mathcal H^2
= \lim_{n \to \infty} \int \Big( |(f \circ h) \, u|^2 
    - \int |(f \circ h) \, u \, \one_{[|h| \leq n]}|^2 \Big)
= 0.
\]
Similarly $\lim_{n \to \infty} \|u - u \, \one_{[|h| \leq n]}\|_\mathcal H^2 = 0$
and $\lim u \, \one_{[|h| \leq n]} = u$ in $D(A)$.
Hence $D(A)$ is a core for $f(A)$.

`\ref{lBS201-3}'.
If $n \in \Ni$ and $u \in D(f(A))$, then 
$|(\chi_n \circ f \circ h) \cdot u|^2 \leq |(f \circ h) \cdot u|^2$
and also $\lim_{n \to \infty}  (\chi_n \circ f \circ h) \cdot u = (f \circ h) \cdot u$.
Then the statement follows from the Lebesgue dominated convergence theorem.
\end{proof}

The next two lemmas are  trivial.

\begin{lemma} \label{lBS202}
Let $A,B$ be operators in a Hilbert space $\mathcal H$ and $T \in \cl(\mathcal H)$.
Let $D$ be a subspace of $\mathcal H$.
Suppose that $D$ is a core for $B$, the operator $A$ is closed 
and $A u = B u + T u$ for all $u \in D$.
Then $D(B) \subset D(A)$ and $A u = B u + T u$ for all $u \in D(B)$.
\end{lemma}
\begin{proof}
Let $u \in D(B)$. 
Then there exists a sequence $(u_n)_{n \in \Ni}$ in $D$ 
such that $\lim u_n = u$ in $D(B)$.
Now $\lim A u_n = \lim (B u_n + T u_n) = B u + T u$.
Since $A$ is closed, it follows that  $u \in D(A)$ and $A u = B u + T u$.
\end{proof}

\begin{lemma} \label{lBS203}
Let $T$ be a densely defined symmetric operator in a Hilbert space $\mathcal H$.
Let $(u_\alpha)_{\alpha \in I}$ be an orthonormal basis in $\mathcal H$.
Suppose that $u_\alpha \in D(T)$ and $\sum_{\alpha \in I} \|T u_\alpha\|_\mathcal H^2 < \infty$.
Then $T$ extends to a bounded operator on $\mathcal H$.
\end{lemma}
\begin{proof}
Let $u \in D(T)$.
Then 
\begin{eqnarray*}
\|T u\|_\mathcal H^2
& = & \sum_{\alpha \in I} |(T u, u_\alpha)|^2
= \sum_{\alpha \in I} |( u, T u_\alpha)|^2
\leq \sum_{\alpha \in I} \|u\|_\mathcal H^2 \, \|T u_\alpha\|_\mathcal H^2
= \|u\|_\mathcal H^2 \, \sum_{\alpha \in I} \|T u_\alpha\|_\mathcal H^2
\end{eqnarray*}
and the lemma follows.
\end{proof}

Now we prove Theorem~\ref{tBS101} in the special case
that $A$ and $B$ have a pure point spectrum.
The proof is very elementary.
For simplicity we assume that $f$ is bounded.

\begin{proposition} \label{pBS204}
Adopt the assumptions and notation of Theorem~{\rm \ref{tBS101}}.
Suppose in addition that $A$ and $B$ both have a pure point spectrum
and that the function $f$ is bounded.
Then $f(A) - f(B)$ is a Hilbert--Schmidt operator and
\[
\|f(A) - f(B)\|_\HS
\leq L \, \|C\|_\HS
.  \]
\end{proposition}
\begin{proof}
Since $A$ has a pure point spectrum there exists an orthonormal basis 
$(u_\alpha)_{\alpha \in I}$ for $\mathcal H$ and for all $\alpha \in I$ there exists 
a $\lambda_\alpha \in \Ri$ such that $A u_\alpha = \lambda_\alpha \, u_\alpha$.
Similarly for $B$ there exists an orthonormal basis 
$(v_\alpha)_{\alpha \in I}$ for $\mathcal H$ and for all $\alpha \in I$ there exists 
a $\mu_\alpha \in \Ri$ such that $B v_\alpha = \mu_\alpha \, v_\alpha$.

If $\alpha,\beta \in I$, then
\begin{eqnarray*}
((f(A) - f(B)) u_\alpha, v_\beta)_\mathcal H
& = & (f(A) u_\alpha, v_\beta)_\mathcal H - ( u_\alpha, f(B) v_\beta)_\mathcal H  \\
& = & (f(\lambda_\alpha) \, u_\alpha, v_\beta)_\mathcal H - ( u_\alpha, f(\mu_\beta) \, v_\beta)_\mathcal H
= (f(\lambda_\alpha) - f(\mu_\beta)) \, (u_\alpha, v_\beta)_\mathcal H
.
\end{eqnarray*}
Now let $\alpha \in I$.
Then Parseval's identity gives
\begin{eqnarray*}
\|(f(A) - f(B)) u_\alpha\|_\mathcal H^2
& = & \sum_{\beta \in I} |((f(A) - f(B)) u_\alpha, v_\beta)_\mathcal H|^2  
= \sum_{\beta \in I} |f(\lambda_\alpha) - f(\mu_\beta) |^2 
       \, |(u_\alpha, v_\beta)_\mathcal H|^2  \\
& \leq & \sum_{\beta \in I} L^2 \, |\lambda_\alpha - \mu_\beta|^2 
       \, |(u_\alpha, v_\beta)_\mathcal H|^2  
= L^2 \, \|(A - B) u_\alpha\|_\mathcal H^2  
= L^2 \, \|C u_\alpha\|_\mathcal H^2
.  
\end{eqnarray*}
Hence 
\[
\|f(A) - f(B)\|_\HS^2
= \sum_{\alpha \in I} \|(f(A) - f(B)) u_\alpha\|_\mathcal H^2
\leq \sum_{\alpha \in I} L^2 \, \|C u_\alpha\|_\mathcal H^2
= L^2 \, \|C\|_\HS^2
< \infty
\]
and the proposition follows.
\end{proof}

We now drop the assumption  that $A$ and $B$ have a pure point spectrum,
but add the assumption that $\ch$ is separable.
The proof is a modification of the proof of \cite{GesztesyPushnitskiSimon} Theorem~4.1.

\begin{proposition} \label{pBS205}
Adopt the assumptions and notation of Theorem~{\rm \ref{tBS101}}.
Suppose in addition that the function $f$ is bounded and the Hilbert space 
$\mathcal H$ is separable.
Then $f(A) - f(B)$ is a Hilbert--Schmidt operator and
\[
\|f(A) - f(B)\|_\HS
\leq L \, \|C\|_\HS
.  \]
\end{proposition}
\begin{proof}
By Proposition~\ref{pBS204} we my assume that $\mathcal H$ is infinite dimensional.
Since $\mathcal H$ is separable and $A$ is self-adjoint 
there exists an orthonormal basis $(e_k)_{k \in \Ni}$
for $\mathcal H$ with $e_k \in D(A)$ for all $k \in \Ni$.
For all $n \in \Ni$ let $P_n$ be the orthogonal projection of $\mathcal H$ onto 
$\spann \{ e_1,\ldots,e_n \} $, set $A_n = P_n \, A \, P_n$ and 
$B_n = P_n \, B \, P_n$.
Now $\lim_{n \to \infty} A_n = A$ in the strong resolvent sense.
Therefore $\lim_{n \to \infty} f(A_n) = f(A)$ strongly in $\mathcal H$ by 
\cite{RS1} Theorem~VIII.20, where we use that $f$ is bounded.
In particular, $\lim_{n \to \infty} f(A_n) e_k = f(A) e_k$ for all $k \in \Ni$.
Similarly $\lim_{n \to \infty} f(B_n) e_k = f(B) e_k$ for all $k \in \Ni$.

Let $m \in \Ni$.
Then for all $n \in \Ni$ with $n > m$ we obtain by Proposition~\ref{pBS204} that 
\begin{eqnarray*}
\sum_{k=1}^m \|(f(A_n) - f(B_n)) e_k\|_\mathcal H^2
& \leq & \sum_{k=1}^n \|(f(A_n) - f(B_n)) e_k\|_\mathcal H^2
= \|f(A_n) - f(B_n)\|_\HS^2  \\
& \leq & L^2 \, \|A_n - B_n\|_\HS^2
= L^2 \, \|P_n \, C \, P_n\|_\HS^2
\leq L^2 \, \|C\|_\HS^2
.  
\end{eqnarray*}
Now take the limit $n \to \infty$ to deduce that 
\[
\sum_{k=1}^m \|(f(A) - f(B)) e_k\|_\mathcal H^2
\leq L^2 \, \|C\|_\HS^2
.  \]
The limit $m \to \infty$ gives
$\|f(A) - f(B)\|_\HS^2 \leq L^2 \, \|C\|_\HS^2$ 
as required.
\end{proof}

In the next step we drop the assumption that $f$ is bounded, 
but keep that $\mathcal H$ is separable.

\begin{proposition} \label{pBS206}
Adopt the assumptions and notation of Theorem~{\rm \ref{tBS101}}.
Suppose in addition that the  Hilbert space 
$\mathcal H$ is separable.
Then $D(A) = D(B) \subset D(f(A)) = D(f(B))$ and the operator 
$f(A) - f(B)$ extends to a Hilbert--Schmidt operator,
by abuse of notation still denoted by $f(A) - f(B)$, such that 
\[
\|f(A) - f(B)\|_\HS
\leq L \, \|C\|_\HS
.  \]
\end{proposition}
\begin{proof}
The equality $D(A) = D(B)$ is trivial.

For all $n \in \Ni$ define $\chi_n \colon \Ri \to \Ri$ by 
$\chi_n = (-n) \vee x \wedge n$.
Then $\chi_n$ is Lipschitz continuous with Lipschitz constant $1$.
Hence the composition $f_n = \chi_n \circ f$ is Lipschitz continuous 
with Lipschitz constant $M$.
Define $T_n = f_n(A) - f_n(B)$.
It follows from Proposition~\ref{pBS205} that $T_n \in \cl_\HS(\mathcal H)$ and 
$\|T_n\|_\HS \leq L \, \|C\|_\HS$.
Hence the sequence $(T_n)_{n \in \Ni}$ is bounded in the Hilbert space
$\cl_\HS(\mathcal H)$.
Passing to a subsequence if necessary, there exists a $T \in \cl_\HS(\mathcal H)$
such that $\lim_{n \to \infty} T_n = T$ weakly in $\cl_\HS(\mathcal H)$.
Then $\lim_{n \to \infty} T_n = T$ in the weak operator topology on $\cl(\mathcal H)$
and $\lim_{n \to \infty} T_n u = T u$ weakly in $\mathcal H$ for all $u \in \mathcal H$

Now let $u \in D(A)$.
Then $u \in D(f(A))$ by Lemma~\ref{lBS201}\ref{lBS201-1}.
So $\lim_{n \to \infty} f_n(A) u = f(A) u$ in $\mathcal H$ by Lemma~\ref{lBS201}\ref{lBS201-3}.
Similarly $\lim_{n \to \infty} f_n(B) u = f(B) u$.
Hence $\lim_{n \to \infty} T_n u = (f(A) - f(B)) u$ in $\mathcal H$.
Therefore $T u = (f(A) - f(B)) u$ for all $u \in D(A)$
and $f(A) u = f(B) u + T u$ for all $u \in D(A)$.
Since $D(A)$ is a core for $f(B)$ by Lemma~\ref{lBS201}\ref{lBS201-2} 
and $f(A)$ is closed, it follows from Lemma~\ref{lBS202} that 
$D(f(B)) \subset D(f(A))$ and 
$f(A) u = f(B) u + T u$ for all $u \in D(f(A))$.
Similarly $D(f(A)) \subset D(f(B))$.
Hence $D(f(A)) = D(f(B))$.
Moreover, $T$ is an extension of $f(A) - f(B)$.

Finally, since $\lim_{n \to \infty} T_n = T$ weakly in $\cl_\HS(\mathcal H)$
one deduces that 
\[
\|T\|_\HS
\leq \liminf_{n \to \infty} \|T_n\|_\HS
\leq L \, \|C\|_\HS
\]
as required.
\end{proof}

Finally we drop the assumption that $\mathcal H$ is separable.

\begin{proof}[{\bf Proof of Theorem~\ref{tBS101}.}]
Since $C$ is a compact operator, there exists a separable subspace $\mathcal H_0$ of 
$\mathcal H$ such that $C(\mathcal H_0) \subset \mathcal H_0$ and $C u = 0$ for all $u \in \mathcal H_0^\perp$.
Let $U$ and $V$ be the Cayley transforms of $A$ and $B$.
There exists a separable closed subspace $\mathcal H_1$ of $\mathcal H$ such that
$\mathcal H_0 \subset \mathcal H_1$ and $\mathcal H_1$ is invariant under the four bounded 
operators $U$, $V$, $U^*$ and $V^*$.
Write $\mathcal H_2 = \mathcal H_1^\perp$.
Then $\mathcal H_2$ is invariant under $U$ and $V$.
Let $P \colon \mathcal H \to \mathcal H_1$ be the orthogonal projection onto~$\mathcal H_1$.
Then the restriction of $U$ to $\mathcal H_1$ is a unitary map from $\mathcal H_1$ onto~$\mathcal H_1$.
This corresponds to a self-adjoint operator $A_1$ in $\mathcal H_1$.
Moreover, if $u \in D(A)$, then $P u \in D(A_1)$ and $A_1 P u = P A u$.
Similarly one can define a self-adjoint operator $A_2$ in~$\mathcal H_2$.
Moreover, one can define similarly operators $B_1$ and~$B_2$.

Since $A = B + C$ one deduces that  $A_1 = B_1 + C$ and $A_2 = B_2$.
Proposition~\ref{pBS206} gives 
$D(A_1) = D(B_1) \subset D(f(A_1)) = D(f(B_1))$ and the operator 
$f(A_1) - f(B_1)$ extends to a Hilbert--Schmidt operator on $\mathcal H_1$,
denoted by $T_1$, such that 
$\|T_1\|_\HS \leq L \, \|C\|_\HS$.
Trivially $f(A_2) = f(B_2)$.
Now $f(A) = f(A_1) \oplus f(A_2)$, with a similar decomposition for $f(B)$.
So $D(f(A)) = D(f(B))$.
Moreover, 
$f(A) - f(B) = (f(A_1) - f(B_1)) \oplus 0$.
Hence the operator $T_1 \oplus 0$ extends $f(A) - f(B)$.
Furthermore, $T \oplus 0$ is a Hilbert--Schmidt operator
with $\|T \oplus 0\|_\HS \leq L \, \|C\|_\HS$.
The proof of Theorem~\ref{tBS101} is complete.
\end{proof}

\section{Resolvent estimates} \label{s-resolvent}
Very often one is in the situation that self-adjoint operators $A$ and $B$ do not 
differ by a Hilbert--Schmidt operator, but the difference of 
the resolvents $(A - \lambda)^{-1} - (B - \lambda)^{-1}$ is a Hilbert--Schmidt operator.
Under an additional decay assumption on the function $f$ we shall 
show that $f(A) - f(B)$ is a Hilbert--Schmidt operator and 
estimate the Hilbert--Schmidt norm in terms of the 
Hilbert--Schmidt norm of the difference of the resolvents.

\begin{theorem} \label{t-BSunbounded}
Let $A$ and $B$ be two lower-bounded self-adjoint operators
with lower bound $\rho \in \Ri$.
Let $\lambda > - \rho$.
Let $f \colon (0,\frac{1}{\lambda - \rho}] \to \Ri$ be a Lipschitz continuous function
with Lipschitz constant $L$.
Define $g \colon [-\rho,\infty) \to \Ri$ by $g(x) = f(\frac{1}{x + \lambda})$.
Suppose that the difference of the resolvents $(A+ \lambda \, I)^{-1} - (B  + \lambda \, I)^{-1}$
is a Hilbert--Schmidt operator.
Then $g(A)$ and $g(B)$ are bounded operators, their difference is a Hilbert--Schmidt operator
and 
\[
\|g(A) - g(B) \|_\HS \leq L \, \|(A + \lambda )^{-1} - (B + \lambda )^{-1} \|_\HS.
\]
\end{theorem}
\begin{proof}
Note that $g(A) = f((A+ \lambda \, I)^{-1})$.
We may extend $f$ to a Lipschitz continuous function, again denoted by $f$, 
by taking it constant on $(-\infty,0]$ and constant on $(\frac{1}{\lambda - \rho},\infty)$.
Then $f$ is bounded, so $f((A+ \lambda \, I)^{-1})$ is bounded.
Similarly $g(B) = f((B+ \lambda \, I)^{-1})$ is bounded.
Then the theorem follows from Theorem~\ref{tBS101} or Proposition~\ref{pBS205}.
\end{proof}

\begin{cor} \label{cBS301}
Let $A$ and $B$ be two lower-bounded self-adjoint operators
with lower bound $\rho \in \Ri$.
Let $\lambda > - \rho$.
Let $g \colon [\rho,\infty) \to \Ri$ be a differentiable function
and suppose that $L = \sup_{x \in [\rho,\infty)} (x+\lambda)^2 \, |g'(x)| < \infty$.
Suppose that the difference of the resolvents $(A+ \lambda \, I)^{-1} - (B  + \lambda \, I)^{-1}$
is a Hilbert--Schmidt operator.
Then $g(A)$ and $g(B)$ are bounded operators, their difference is a Hilbert--Schmidt operator
and 
\[
\|g(A) - g(B) \|_\HS \leq L \, \|(A + \lambda )^{-1} - (B + \lambda )^{-1} \|_\HS
.  \]
\end{cor}
\begin{proof}
Define $f \colon (0,\frac{1}{\lambda - \rho}] \to \Ri$ by $f(t) = g(\frac{1}{t} - \lambda)$.
Then $g(x) = f(\frac{1}{x + \lambda})$ for all $x \in [\rho,\infty)$.
Moreover, $f'(t) = - \frac{1}{t^2} \, g'(\frac{1}{t} - \lambda) = - (x+\lambda)^2 \, g'(x)$
for all $t \in (0,\frac{1}{\lambda - \rho}]$, where $x = \frac{1}{t} - \lambda$.
Hence $f$ is Lipschitz continuous with Lipschitz constant~$L$.
Now apply Theorem~\ref{t-BSunbounded}.
\end{proof}

\section{A contractive iteration scheme} \label{ss-scheme}

In this section we provide a sufficient condition to obtain a solution of 
a nonlinear map on a real Hilbert space via the Banach contraction theorem.

If $\ch$ is a real Hilbert space, $A \colon \ch \to \ch^*$ is a map
and $m > 0$, then $A$ is called {\bf strongly monotone} with {\bf monotonicity constant} $m$
if $\langle Au - Av, u-v \rangle_{\ch^* \times \ch} \geq m \, \|u-v\|_\ch^2$
for all $u,v \in \ch$.
Note that then $\|A u - A v\|_\ch \geq m \|u-v\|_\ch$ and $A$ is injective.

\begin{theorem} \label{t-ggz}
Let $\mathcal H$ be a real Hilbert space, $m > 0$ and 
$A\colon \mathcal H \to \mathcal H^*$ a strongly monotonous map
with monotonicity constant $m$.
Let $R > 0$ and let $\cb = \{ u \in \ch : \|u\|_\ch \leq R \} $ 
be the closed ball centred at $0$ and radius $R$.
Suppose that the restriction 
$A|_\cb$ is Lipschitz continuous with Lipschitz constant $M$.
Let $J\colon \mathcal H \to \mathcal H^*$ denote the duality map.
Fix $f \in \ch^*$.
Define $Q \colon \cb \to \ch$ by
\[
Qu = u - \frac {m}{M^2} \, J^{-1}(Au -f )
.  \]
Then one has the following.
\begin{tabel} 
\item \label{t-ggz-1}
The map $Q$ is a contraction with contraction constant $\sqrt{ 1 - \frac {m^2}{M^2}}$.
\item \label{t-ggz-2}
Suppose $R \ge \frac {2}{m} \, \|A0-f\|_{\mathcal H^*}$.
Then $\Q$ maps $\mathcal B$ into itself.
\item \label{t-ggz-3}
Suppose $R \ge \frac {2}{m} \, \|A0-f\|_{\mathcal H^*}$.
Then there exists a unique $u \in \cb$ such that 
$A u = f$.
Define $u_1 = 0$ and for all $n \in \Ni$ define $u_{n+1} = Q u_n$.
Then $u = \lim_{n \to \infty} u_n$.
Moreover,
\[
\|u\|_{\mathcal H} \leq \frac {1}{m} \, \|A0-f\|_{\mathcal H^*}.
\]
\item \label{t-ggz-4}
Let also $g \in \ch^*$.
Suppose $R \ge \frac {2}{m} \, \|A0-f\|_{\mathcal H^*}$
and $R \ge \frac {2}{m} \, \|A0-g\|_{\mathcal H^*}$.
Then there exist unique $u,v \in \ch$ such that 
$A u = f$ and $A v = g$.
Furthermore, $\|u - v\|_\ch \leq \frac{1}{m} \, \|f-g\|_{\ch^*}$.
\end{tabel}
\end{theorem}
\begin{proof}
`\ref{t-ggz-1}'.
Let $u,v \in \cb$.
Using the duality map $J$, one estimates
\begin{eqnarray*}
\|Q u - Q v\|_\mathcal H^2 
& = & (u - v -\frac {m}{M^2} \, J^{-1}(Au  - Av) , 
            u - v -\frac {m}{M^2}J^{-1}(Au  - Av))_\mathcal H  \\
& = & \|u-v\|_\mathcal H^2 -2\frac {m}{M^2} (J^{-1}(Au  - Av),u-v)_\mathcal H + 
   \frac {m^2}{M^4} \, \|J^{-1}(Au  - Av)\|^2_\mathcal H  \\
& = & \|u-v\|_\mathcal H^2 -2\frac {m}{M^2} 
     \langle Au  - Av, u-v \rangle _{\mathcal H ^*\times \mathcal H }
   + \frac {m^2}{M^4} \, \|Au  - Av\|^2_{\mathcal H^*}  \\
& \le & \|u-v\|_\mathcal H^2 -2\frac {m^2}{M^2} \|u-v \|^2 _{\mathcal H }+ 
\frac {m^2}{M^2} \|u  - v\|^2_{\mathcal H} 
= (1- \frac {m^2}{M^2}) \|u  - v\|^2_{\mathcal H}
,  
\end{eqnarray*}
where we used the strong monotonicity and Lipschitz continuity in the 
last inequality.

`\ref{t-ggz-2}'.
Let $u \in \mathcal B$.
Then if follows from Statement~\ref{t-ggz-1} that 
\begin{eqnarray*}
\|\Q u\| _\mathcal H 
& \le & \|\Q u -\Q 0\| _\mathcal H +\|\Q 0\| _\mathcal H 
\le \sqrt{ 1 - \frac {m^2}{M^2}}\|u\|_\mathcal H + \frac {m}{M^2}\|A0 -f \|_{\mathcal H^*}  \\
& \le & \Bigl(\sqrt{ 1 - \frac {m^2}{M^2}} +  \frac {m^2}{2M^2}\Bigr ) R 
\le R
.
\end{eqnarray*}

`\ref{t-ggz-3}'. 
Most follows from the Banach fixed point theorem. 
Since
\[
\|f- A0 \|_{\mathcal H^*} \, \|u\|_\mathcal H 
= \|Au - A0 \|_{\mathcal H^*} \, \|u\|_\mathcal H 
\ge \langle Au - A0 ,u -0 \rangle_{\mathcal H ^*\times \mathcal H } 
\ge m \|u\|_\ch^2,
\]
the estimate follows.

`\ref{t-ggz-4}'. 
The existence follows from Statement~\ref{t-ggz-3}.
The monotonicity gives
$m \, \|u-v\|_\ch^2
\leq \langle Au - Av, u-v \rangle_{\ch^* \times \ch}
= \langle f-g, u-v \rangle_{\ch^* \times \ch}
\leq \|f-g\|_{\ch^*} \, \|u-v\|_\ch$
and hence $m \, \|u-v\|_\ch \leq	 \|f-g\|_{\ch^*}$.
\end{proof}

In \cite{ZeidlerIIB} Section~25.4
and \cite{GGZ} Theorem~3.4 in Subsection~III.3.2  \emph{global} Lipschitz continuity 
for $A$ is demanded, but in the latter is at least indicated that local 
Lipschitz continuity suffices. 
The required radius of the ball $\mathcal B$ in \cite{GGZ}, however, is not optimal, 
and the proof is unnecessarily complicated.
In general, \emph{global} Lipschitz continuity is a severe restriction and, in 
general not fulfilled if $A$ is not linear.
In particular, the 
operator used in Proposition~\ref{tBS509.5} is not globally Lipschitz.

\section{An application in density functional theory} \label{s-SchroedPois}

In this section we give an application of Corollary~\ref{cBS301}
via a Schr\"odinger--Poisson system used in electron transport theory.
The main theorem of this section is as follows.

\begin{theorem} \label{tBS510}
Let $d \in \{ 1,2,3 \} $ and $\Omega \subset \Ri^d$ a bounded open set 
with Lipschitz boundary.
Let $\varepsilon, \mathfrak m \colon \Omega \to \Ri^{d \times d}$ be bounded
measurable functions with $\varepsilon(x)$ and $\mathfrak m(x)$
symmetric for all $x \in \Omega$.
Let $\mu, \tilde \mu > 0$.
Suppose that  
\[
(\varepsilon(x) \xi, \xi)_{\Ri^d} \geq \mu \, \|\xi\|_{\Ri^d}^2
\quad \mbox{and} \quad 
(\gotm(x) \xi, \xi)_{\Ri^d} \geq \tilde \mu \, \|\xi\|_{\Ri^d}^2
\]
for all $\xi \in \Ri^d$.
Let $f\colon \R \to (0,\infty)$ be a differentiable and strictly 
decreasing function with $f'$ locally bounded.
Moreover, suppose that both $r \mapsto r^4 \, f(r)$ 
and $r \mapsto r^4 \, f'(r)$ are bounded on $[0,\infty)$.
Then $f \colon \Ri \to f(\Ri)$ is bijective and assume 
in addition that its inverse is locally Lipschitz.
Let $D \subset \partial \Omega$ be closed.
Fix $N > 0$.

Let $W^{1,2}_D(\Omega)$ be the $W^{1,2}(\Omega)$-closure of the set 
$ \{ u|_\Omega : u \in C^\infty(\R^d) \mbox{ and } \supp u \cap D = \emptyset \} $.
Let $c_P > 0$.
Suppose that 
\[
\|u\|_{L_2(\Omega)}^2 \leq c_P \, \|\nabla u\|_{L_2(\Omega)}^2
\]
for all $u \in W^{1,2}_D(\Omega)$.
(So the space $W^{1,2}_D(\Omega)$ admits a Poincar\'e inequality.)
For all $V \in L_2(\Omega,\Ri)$ and $\psi,\varphi \in W^{1,2}_D(\Omega)$
it follows that $V \, \psi \, \overline \varphi \in L_1(\Omega)$.
Define $\gott_V \colon W^{1,2}_D(\Omega) \times W^{1,2}_D(\Omega) \to \Ci$ by 
\[
\gott_V[\psi,\varphi] 
= \gott[\psi,\varphi] + \int_\Omega V \, \psi \, \overline \varphi
.  \]
Then one has the following.
\begin{tabel}
\item \label{tBS510-1}
For all $V \in L_2(\Omega,\Ri)$ the sesquilinear form
$\gott_V$ is closed, symmetric and lower-bounded.
We denote by $H$ the associated operator in case $V = 0$ and by
$H \dotplus V$ the associated lower-bounded 
self-adjoint operator.
\item \label{tBS510-2}
For all $V \in L_2(\Omega,\Ri)$ the operator $f(H \dotplus V)$ is nuclear.
\item \label{tBS510-3}
For all $V \in L_2(\Omega,\Ri)$ there exists a unique $\ce(V) \in \Ri$ 
such that 
\[
\tr f(H \dotplus (V - \ce(V) \, \one_\Omega)) = N
.  \]
\item \label{tBS510-4}
For all $V \in L_2(\Omega,\Ri)$ there exists a unique $\cn(V) \in L_2(\Omega,\Ri)$ 
such that 
\[
\int_\Omega \cn(V) \, W 
= \tr \bigl ( M_W \, f(H \dotplus V - \ce(V)) \bigr )
\]
for all $W \in L_\infty(\Omega)$, where $M_W$ is the multiplication operator with~$W$.
\item \label{tBS510-5}
For all $V_0 \in L_2(\Omega,\Ri)$ and $q \in (W^{1,2}_D(\Omega,\Ri))^*$
there exists a unique $V \in W^{1,2}_D(\Omega,\Ri)$
such that 
\[
\int_\Omega \varepsilon \nabla V \cdot \nabla W
   - (\cn(V_0 + V), W)_{L_2(\Omega)}
= \langle q,W \rangle_{(W^{1,2}_D(\Omega,\Ri))^* \times W^{1,2}_D(\Omega,\Ri)}
\]
for all $W \in W^{1,2}_D(\Omega,\Ri)$.
Write $\Psi(V_0,q) = V$.
\item \label{tBS510-6}
If $V_0 \in L_2(\Omega,\Ri)$ and $q \in (W^{1,2}_D(\Omega,\Ri))^*$,
then 
\[
\|\Psi(V_0,q)\|_{W^{1,2}_D(\Omega,\Ri)}
\leq \frac{1 + c_P}{\mu} \, \|q + \cn(V_0)\|_{(W^{1,2}_D(\Omega))^*}
.  \]
\item \label{tBS510-7}
If $V_0 \in L_2(\Omega,\Ri)$,
then 
\[
\|\Psi(V_0,q) - \Psi(V_0,\tilde q)\|_{W^{1,2}_D(\Omega)}
\leq \frac{1 + c_P}{\mu} \, \|q - \tilde q\|_{(W^{1,2}_D(\Omega))^*}
\]
for all $q,\tilde q \in (W^{1,2}_D(\Omega,\Ri))^*$.
\item \label{tBS510-8}
For all $R > 0$ there is a $c > 0$ such that
\[
\|\Psi(V_0,q) - \Psi(V_1,\tilde q)\|_{W^{1,2}_D(\Omega)}
\leq c \, \|V_0 - V_1\|_{L_2}
   + \frac{1 + c_P}{\mu} \, \|q - \tilde q\|_{(W^{1,2}_D(\Omega))^*}
\]
for all $V_0,V_1 \in L_2(\Omega,\Ri)$ and $q,\tilde q \in (W^{1,2}_D(\Omega,\Ri))^*$ with 
$\|V_0\|_{L_2} \leq R$, $\|V_1\|_{L_2} \leq R$, 
$\|q\|_{(W^{1,2}_D(\Omega))^*} \leq R$ and $\|\tilde q\|_{(W^{1,2}_D(\Omega))^*} \leq R$.
\end{tabel}
\end{theorem}

The proof of Theorem~\ref{tBS510} requires quite some preparation
and the statements will be proved in separate lemmas and propositions in 
this section.
Throughout this section we adopt the notation and assumptions of Theorem~\ref{tBS510}.

Define the sesquilinear form 
$\mathfrak t \colon W^{1,2}_D (\Omega) \times W^{1,2}_D (\Omega) \to \Ci$
by
\[
\mathfrak t[\psi,\varphi] 
=  \int_\Omega \mathfrak  m^{-1} \nabla \psi \cdot \overline{\nabla \varphi} .
\]
Then $\gott$ is a closed positive symmetric form.
We define the unperturbed Schr\"odinger operator $H$ 
on $L_2(\Omega) $ as the operator associated to the form $\gott$.
Formally $H=-\dive \bigl (\mathfrak   m^{-1}\nabla \bigr )$
In semiconductor modelling $\mathfrak m$ describes the position dependent effective mass
(modulo Planck's constant and a factor~$2$).
Then $H$ is a positive self-adjoint operator.
The form domain $W^{1,2}_D(\Omega)$  
gives the realisation of the operator
$H$ with homogeneous Dirichlet conditions on $D$ and
homogeneous Neumann conditions on
$\partial \Omega \setminus D$.
We refer to \cite{EgertTolksdorf} for more details of the spaces $W^{1,2}_D (\Omega)$.

We need for the sequel that the resolvent of $H$ is a Hilbert--Schmidt operator,
or equivalently that $(H +1)^{-1/2}$ belongs to the Schatten $4$-class.

\begin{proposition} \label{p-HSWeyl}
$(H +1)^{-1/2} \in \cs_4$.
\end{proposition}
\begin{proof}
Define $\gotl \colon W^{1,2}(\Omega) \times W^{1,2}(\Omega) \to \Ci$ by
$\gotl[\psi,\varphi] = \int_\Omega \nabla \psi \cdot \overline{\nabla \varphi}$.
Then there is a $c > 0$ such that 
$\gotl[\psi] \leq c \, \gott[\psi]$ for all $\psi \in W^{1,2}_D(\Omega)$.
If $\Delta_N$ is the Laplacian with Neumann boundary conditions, 
then $- \Delta_N$ is the operator associated with $\gotl$.
Let $\lambda_1 \leq \lambda_2 \leq \ldots$ be the eigenvalues of the operator $H$,
repeated with multiplicity and let 
$\mu_1 \leq \mu_2 \leq \ldots$ be the eigenvalues of the operator $-\Delta_N$,
repeated with multiplicity.
Then the mini-max theorem gives $\mu_n \leq c \, \lambda_n$ for all $n \in \Ni$.
The Weyl asymptotics give that there is a $c' > 0$ such that $\mu_n \geq c' \, n^{2/d}$
for large $n \in \Ni$.
Since $d \leq 3$ this implies that $\sum_{n=1}^\infty \Big( (\lambda_n + 1)^{-1/2} \Big)^4 < \infty$.
\end{proof}

In the sequel $H$ is perturbed by a potential $V \in L_2(\Omega,\Ri)$.
We will prove a form bound
with respect to $\mathfrak t$ in order to see that the form sum is well-defined.

\begin{lemma} \label{l-relformbound}
There is a $\gamma > 0$ such that 
\[
\Big | \int_\Omega V \psi \, \overline \psi \Big |
\leq   \frac34( \mathfrak t +1) [\psi]   +  \gamma \|V\|_{L_2}^4 \|\psi \|^2_{L_2}
\]
uniformly for all $V \in L_2(\Omega,\Ri)$ and $\psi \in W^{1,2}_D(\Omega)$.
\end{lemma}
\begin{proof}
Since $\Omega$ is Lipschitz, one has the continuous embedding
$W^{1,2}_D(\Omega) \hookrightarrow L_6(\Omega)$, where we use that $d \leq 3$. 
Let $c_1$ be the embedding constant.
Then 
\begin{eqnarray*}
\Big | \int_\Omega V \psi \, \overline \psi \Big |
& \leq & \|V \|_{L_2} \|\psi \|^2 _{L_4} 
\leq  \|V \|_{L_2} ( \|\psi \|^2_{L_6}) ^{\frac34}
 ( \|\psi \|^2_{L_2}) ^{\frac14} 
\leq 
c_1^{\frac{3}{2}} \, \|V \|_{L_2} ( \|\psi \|^2_{W^{1,2}_D}) ^{\frac34} ( \|\psi \|^2_{L_2}) ^{\frac14}   \\
& \leq & c_1^{\frac{3}{2}} \Big( (\|\gotm\|_{L_\infty} + 1) (\gott + 1)[\psi] \Big)^{\frac{3}{4}}
   \|V\|_{L_2} ( \|\psi \|^2_{L_2})^{\frac14}   \\
& \leq & \frac{3}{4} (\gott + 1)[\psi]
   + \frac{ c_1^6 \|V\|_{L_2}^4 (\|\gotm\|_{L_\infty} + 1)^3 }{4} \|\psi \|_{L_2}^2
\end{eqnarray*}
by Young's inequality.
\end{proof}

Throughout the remainder of this section we let $\gamma > 0$ be as in Lemma~\ref{l-relformbound}.
For all $V \in L_2(\Omega,\Ri)$ define the sesquilinear form
$\gott_V \colon W^{1,2}_D (\Omega) \times W^{1,2}_D (\Omega) \to \Ci$
by
\[
\gott_V[\psi,\varphi] 
= \gott[\psi,\varphi] + \int_\Omega V \, \psi \, \overline \varphi
.  \]
By Lemma~\ref{l-relformbound} the form $\gott_V$ is symmetric, semibounded and closed.
We denote the corresponding self-adjoint operator by $H \dotplus V$.

\begin{lemma} \label{l-formabsch}
Let $R > 0$.
Define $\lambda = 1 + \gamma \, R^4$.
Let $V \in L_2(\Omega,\Ri)$ with $\|V\|_{L_2} \leq R$.
Then one has the following.
\begin{tabel}
\item \label{l-formabsch-1}
The operator $H \dotplus V + \lambda$ is positive, invertible and
$\|(H \dotplus V + \lambda)^{-1}\|_{L_2 \to L_2} \leq 4$.
\item \label{l-formabsch-5}
$\frac{1}{4} \, \gott[\psi] - \lambda \, \|\psi\|_{L_2}^2
\leq \gott_V[\psi]
\leq \frac{7}{4} \, \gott[\psi] + \lambda \, \|\psi\|_{L_2}^2$
for all $\psi \in W^{1,2}_D(\Omega)$.
\item \label{l-formabsch-2}
$\| (H +1)^{1/2} (H \dotplus V + \lambda) ^{-1/2}\| \leq 2$.
\item \label{l-formabsch-4}
$\|(H \dotplus V + \lambda) ^{-2}\|_{\cs_1}^{1/4}
= \|(H \dotplus V + \lambda) ^{-1/2}\|_{\cs_4} 
\leq 2 \|(H + 1) ^{-1/2}\|_{\cs_4} < \infty$.
\end{tabel}
\end{lemma}
\begin{proof}
`\ref{l-formabsch-1}'.
Let $\psi \in W^{1,2}_D(\Omega)$.
Then
\[
\mathfrak t_V [\psi] 
= \mathfrak t[\psi] +  \int_\Omega V \psi \, \overline \psi  
\geq \mathfrak t[\psi] - \Big |  \int_\Omega V \psi \, \overline \psi \Big |  
\geq \frac14 (\mathfrak t +1)[\psi] - \lambda \, \|\psi \|^2_{L_2},
\]
or, equivalently,
\begin{equation} \label{e-lowerformbound1}
\mathfrak t_V [\psi] + \lambda \, \|\psi \|^2_{L_2} \geq \frac14 (\mathfrak t+1) [\psi]
.
\end{equation}
This implies Statement~\ref{l-formabsch-1}.

`\ref{l-formabsch-5}'.
This follows similarly.

`\ref{l-formabsch-2}'.
It follows from (\ref{e-lowerformbound1}) that
\[
\|(H \dotplus V + \lambda) ^{1/2} \psi\|_{L_2}^2
\geq \frac{1}{2} \, \|(H +1)^{1/2} \psi\|_{L_2}^2
\]
for all $\psi \in W^{1,2}_D(\Omega)$.
Now use \cite{Kat1} Theorem~VI.2.23.

`\ref{l-formabsch-4}'.
Since $(H \dotplus V + \lambda) ^{-1/2} 
= (H + 1)^{-\frac{1}{2}} \, \Big( (H + 1)^{\frac{1}{2}} (H \dotplus V + \lambda) ^{-1/2} \Big)$,
the statement follows from Statement~\ref{l-formabsch-2} and Proposition~\ref{p-HSWeyl}.
\end{proof}

The form method gives that $D(H^{1/2}) = D(H \dotplus V + 1 + \gamma \, \|V\|_{L_2}^4)^{1/2}$
for all $V \in L_2(\Omega,\Ri)$.
Even more is valid.
For all $V \in L_2(\Omega,\Ri)$ let $M_V$ be the self-adjoint operator in $L_2(\Omega)$ 
by multiplication with $V$.

\begin{prop} \label{pBS404}
Let $R > 0$.
Define $\lambda = 1 + \gamma \, R^4$.
Then there is a $c > 0$ such that for all $V \in L_2(\Omega,\Ri)$ with $\|V\|_{L_2} \leq R$
the following is valid.
\begin{tabel}
\item \label{pBS404-1}
$D(H \dotplus V) = D(H)$.
\item \label{pBS404-2}
$D(H) \subset L_\infty(\Omega) \subset D(M_V)$.
\item \label{pBS404-3}
$H \dotplus V = H + M_V$.
\item \label{pBS404-4}
$\|(H \dotplus V + \lambda)^{-1}\|_{L_2 \to L_\infty} \leq c$.
\end{tabel}
\end{prop}
\begin{proof}
Let $S$ be the semigroup on $L_2$ generated by $-(H+\lambda)$.
It follows from \cite{AE1} Theorem~4.4 together with Example~4.3 that 
there are $c,\omega > 0$ such that 
$\|S_t\|_{L_2 \to L_\infty} \leq c \, t^{-d/4} \, e^{\omega t}$
for all $t > 0$.
Let $\rho \in \Ri$ with $\rho > \omega$
be such that $c \, R \, (\rho - \omega)^{-\frac{4-d}{4}} \, \Gamma(\frac{4-d}{4}) = \frac{1}{2}$.
Let $\psi \in L_2$.
Then a Laplace transform gives $(H + \lambda + \rho)^{-1} \psi \in L_\infty$,
where we use that $d \leq 3$.
In particular, $D(H) \subset D(M_V)$.
Moreover,
\begin{eqnarray*}
\|M_V \, (H + \lambda + \rho)^{-1} \psi\|_{L_2} 
& \leq & \|V\|_{L_2} \, 
   \int_0^\infty e^{-\rho t} \, \|S_t\|_{L_2 \to L_\infty} \, dt \, \|\psi\|_{L_2}  \\
& \leq & c \, \|V\|_{L_2} \, (\rho - \omega)^{-\frac{4-d}{4}} \, \Gamma(\frac{4-d}{4}) \, \|\psi\|_{L_2}  \\
& \leq & \frac{1}{2} \, \|\psi\|_{L_2}
. 
\end{eqnarray*}
Therefore $M_V$ is $H$-bounded with 
relative bound $\frac{1}{2}$.
Then it follows from \cite{Kat1} Theorem~V.4.3 that the operator 
$H + M_V$ is self-adjoint.
Note that $D(H + M_V) = D(H)$.
Let $\psi \in D(H)$.
Then $\psi \in D(M_V)$.
If $\varphi \in W^{1,2}_D (\Omega)$, 
then 
$\gott_V[\psi,\varphi] 
= \gott[\psi,\varphi] + (V \psi,\varphi)_{L_2}
= (H \psi + V \psi, \varphi)_{L_2}$.
Hence $\psi \in D(H \dotplus V)$ and 
$(H \dotplus V) \psi = H \psi + V \psi = (H + M_V) \psi$.
We proved that $H \dotplus V$ is an extension of $H + M_V$.
But both $H \dotplus V$ and $H + M_V$ are self-adjoint.
Therefore $H \dotplus V = H + M_V$.
In particular $D(H \dotplus V) = D(H) \subset L_\infty$.

Next 
\[
(H \dotplus V + \lambda + \rho)^{-1}
   \Big( 1 + M_V \, (H + \lambda + \rho)^{-1} \Big)
= (H + \lambda + \rho)^{-1}
\]
and 
\[
(H \dotplus V + \lambda + \rho)^{-1}
= (H + \lambda + \rho)^{-1} \, 
      \Big( 1 + M_V \, (H + \lambda + \rho)^{-1} \Big)^{-1}
.  \]
Therefore $\|(H \dotplus V + \lambda + \rho)^{-1}\|_{L_2 \to L_\infty}
\leq 2 \|(H + \lambda + \rho)^{-1}\|_{L_2 \to L_\infty}$.
\end{proof}

As a result we obtain a resolvent equation.
Note that the multiplication operator in the next corollary
is an unbounded operator in general.

\begin{cor} \label{cBS405}
Let $R > 0$.
Define $\lambda = 1 + \gamma \, R^4$.
Let $U,V \in L_2(\Omega,\Ri)$ with $\|U\|_{L_2} \leq R$ and $\|V\|_{L_2} \leq R$.
Then $D(H \dotplus V + \lambda) \subset D(M_U)$ and 
\[
(H \dotplus U + \lambda)^{-1} - (H \dotplus V + \lambda)^{-1}
= - (H \dotplus U + \lambda)^{-1} \, M_{U - V} \, (H \dotplus V + \lambda)^{-1}
.  \]
\end{cor}

\begin{lemma} \label{lBS402}
Let $R > 0$.
Define $\lambda = 1 + \gamma \, R^4$.
Then one has the following.
\begin{tabel}
\item \label{lBS402-1}
For all $k \in \{ 0,\ldots,4 \} $
there is a $c_k > 0$ such that 
$f(H \dotplus V) \psi \in D((H \dotplus V + \lambda)^k)$ for all $\psi \in L_2(\Omega)$
and 
$\| (H \dotplus V + \lambda)^k \, f(H \dotplus V)\|_{L_2 \to L_2} \leq c_k$
for all $V \in L_2(\Omega,\Ri)$ with $\|V\|_{L_2} \leq R$.
\item \label{lBS402-2}
For all $k \in \{ 0,1,2 \} $ there exists an $M_k > 0$ such that 
the operator \[
(H \dotplus V + \lambda)^k \, f(H \dotplus V)
\] 
is nuclear and 
$\|(H \dotplus V + \lambda)^k \, f(H \dotplus V)\|_{\cs_1} \leq M_k$
for all $V \in L_2(\Omega,\Ri)$ with $\|V\|_{L_2} \leq R$.
\end{tabel}
\end{lemma}
\begin{proof}
`\ref{lBS402-1}'.
The operator $H \dotplus V$ is self-adjoint and lower bounded by $- \gamma$
by Lemma~\ref{l-formabsch}\ref{l-formabsch-1}.
Define $g \colon \Ri \to \Ri$ by $g(r) = (r+\lambda)^k \, f(r)$.
Then $g$ is bounded on $[-\lambda,\infty)$ by the assumptions on~$f$.
Therefore the spectral theorem gives
$\|(H \dotplus V + \lambda)^k \, f(H \dotplus V)\|_{L_2 \to L_2}
\leq \sup_{r \in [- \lambda,\infty)} |g(r)|
= \sup_{r \in [- \lambda,\infty)} (r + \lambda)^k \, f(r)$.

`\ref{lBS402-2}'.
Now
\begin{eqnarray*}
\|(H \dotplus V + \lambda)^k \, f(H \dotplus V)\|_{\cs_1} 
& \leq & \|\bigl (H \dotplus V + \lambda)^{-2} \|_{\cs_1} \, 
     \|\bigl(  (H \dotplus V + \lambda)^{k+2} \, f(H \dotplus V) \|_{L_2 \to L_2}  \\
& \leq & (2 \|(H+1)^{-1/2}\|_{\cs_4})^4 \, c_{k+2}
,
\end{eqnarray*}
where we used Lemma~\ref{l-formabsch}\ref{l-formabsch-4} and Statement~\ref{lBS402-1}.
\end{proof}

For all $W \in L_\infty(\Omega)$ we denote by $M_W$ the multiplication operator
on $L_2(\Omega)$ by multiplying with $W$.
Let $V \in L_2(\Omega,\Ri)$.
Choosing $k = 0$ in Lemma~\ref{lBS402}\ref{lBS402-2} 
gives that $f(H \dotplus V)$
is trace class.
Then 
\[
|\Tr( M_W f(H \dotplus V) )| 
\leq \|M_W\|_{L_2 \to L_2} \, \Tr |f(H \dotplus V)|
\leq \|W\|_{L_\infty(\Omega)} \, \Tr f(H \dotplus V)
\]
for all $W \in L_\infty(\Omega)$.
Hence there is a unique $\cm(V) \in L_1(\Omega,\Ri)$ such that 
\[
\int_\Omega \mathcal M(V) \, W 
= \tr \bigl ( M_W \, f(H \dotplus V ) \bigr )
\]
for all $W \in L_\infty(\Omega)$.

\begin{prop} \label{p-lipschpartdens}
Let $R > 0$. 
Then there exists an $M > 0$ such that the following is valid.
\begin{tabel}
\item \label{p-lipschpartdens-1}
If $V \in L_2(\Omega,\Ri)$ with $\|V\|_{L_2} \leq R$, then
$\cm(V) \in L_2(\Omega,\Ri)$ and
$\|\cm(V)\|_{L_2} \leq M$.
\item \label{p-lipschpartdens-2}
$\|\M(U) - \M(V) \|_{L_2} \leq M \, \|U - V \|_{L_2}$
for all $U,V \in L_2(\Omega,\Ri)$ with $\|U\|_{L_2} \leq R$
and $\|V\|_{L_2} \leq R$.
\end{tabel}
\end{prop}

\begin{proof}
`\ref{p-lipschpartdens-1}'.
Let $c_3 > 0$ be as in 
Lemma~\ref{lBS402}\ref{lBS402-1}
and $c > 0$ as in Proposition~\ref{pBS404}\ref{pBS404-4}.
Write $\lambda = 1 + \gamma \, R^4$.
Let $W \in L_\infty(\Omega)$.
Then 
\begin{eqnarray*}
\Big| \int_\Omega \mathcal M(V) \, W \Big|
& = & \Big| \tr \bigl( M_W \, f(H \dotplus V) \bigr) \Big|  \\
& = & \Big| \tr \Big( M_W \, (H \dotplus V + \lambda)^{-1}
    (H \dotplus V + \lambda)^3 \, f(H \dotplus V) \, (H \dotplus V + \lambda)^{-2} \Big) \Big|  \\
& \leq & \|M_W \, (H \dotplus V + \lambda)^{-1}\|_{L_2 \to L_2} \, 
     \cdot  \\*
& & \hspace*{5mm} {} \cdot
   \|(H \dotplus V + \lambda)^3 \, f(H \dotplus V)\|_{L_2 \to L_2} \, 
   \|(H \dotplus V + \lambda)^{-2}\|_{\cs_1}  \\
& \leq & c \, \|W\|_{L_2} \, c_3
   \, (2 \|(H + 1)^{-1/2}\|_{\cs_4})^4
,
\end{eqnarray*}
where we used  
Proposition~\ref{pBS404}\ref{pBS404-4}, 
Lemma~\ref{lBS402}\ref{lBS402-1} and Lemma~\ref{l-formabsch}\ref{l-formabsch-4}.

`\ref{p-lipschpartdens-2}'.
Write $\lambda = 1 + \gamma \, R^4$.
Define the function $g \colon \Ri \to \Ri$ by $g(r) = (r + \lambda)^2 \, f(r)$.
Let $U,V \in L_2(\Omega,\Ri)$ with $\|U\|_{L_2} \leq R$
and $\|V\|_{L_2} \leq R$.
Let $W \in L_\infty(\Omega)$.
Then 
\begin{eqnarray*}
\int_\Omega \cm(V) \, W
& = & \Tr (M_W \, f(H \dotplus V) ) \\
& = & \Tr \Big( g(H \dotplus V) \, (H \dotplus V + \lambda)^{-1} \, M_W \, (H \dotplus V + \lambda)^{-1} \Big)  
\end{eqnarray*}
with a similar expression for $U$ instead of $V$.
Therefore
\begin{eqnarray*}
\lefteqn{
\Big| \int_\Omega \Big( \cm(U) - \cm(V) \Big) \, W \Big|
} \hspace*{10mm}  \\*
& = & \Big| \Tr \Big( g(H \dotplus U) - g(H \dotplus V) \Big)
    (H \dotplus U + \lambda)^{-1} \, M_W \, (H \dotplus U + \lambda)^{-1}  \\*
& & \hspace*{10mm} {} + 
    \Tr g(H \dotplus V) \, \Big( (H \dotplus U + \lambda)^{-1} \, M_W \, (H \dotplus U + \lambda)^{-1}  \\*
& & \hspace*{70mm} {} 
         - (H \dotplus V + \lambda)^{-1} \, M_W \, (H \dotplus V + \lambda)^{-1} \Big) \Big|  \\
& \leq & \|g(H \dotplus U) - g(H \dotplus V)\|_\HS \, 
     \|(H \dotplus U + \lambda)^{-1} \, M_W \, (H \dotplus U + \lambda)^{-1}\|_\HS  \\*
& & \hspace*{1mm} {} + 
      \|g(H \dotplus V)\|_\HS \, 
    \|(H \dotplus U + \lambda)^{-1} \, M_W \, (H \dotplus U + \lambda)^{-1}  \\*
& & \hspace*{70mm} {} 
         - (H \dotplus V + \lambda)^{-1} \, M_W \, (H \dotplus V + \lambda)^{-1}\|_\HS
.
\end{eqnarray*}
We estimate the two factors of the two terms.

First consider $\|g(H \dotplus U) - g(H \dotplus V)\|_\HS$.
If $r \in [-\lambda,\infty)$, then 
\[
(r + \lambda)^2 \, |g'(r)|
= 2 (r + \lambda)^3 \, f(r) + (r + \lambda)^4 \, f'(r)
.  \]
Set 
\[
L = \sup_{r \in [-\lambda,\infty)} (r + \lambda)^2 \, |g'(r)|
.  \]
The assumptions on the function~$f$ imply that $L < \infty$.
Then Corollary~\ref{cBS301} together with Corollary~\ref{cBS405} give
\begin{eqnarray*}
\lefteqn{
\|g(H \dotplus U) - g(H \dotplus V)\|_\HS
} \hspace*{10mm}  \\*
& \leq & L \, \|(H \dotplus U + \lambda)^{-1} - (H \dotplus V + \lambda)^{-1}\|_\HS  \\
& = & L \, \|(H \dotplus U + \lambda)^{-1} \, M_{U - V} \, (H \dotplus V + \lambda)^{-1}\|_\HS  \\
& \leq & \|(H \dotplus U + \lambda)^{-1}\|_\HS \, 
   \|M_{U - V} \, (H \dotplus V + \lambda)^{-1}\|_{L_2 \to L_2}   \\
& \leq & (2 \|(H + 1)^{-1/2}\|_{\cs_4})^2  \, c \, \|U - V\|_{L_2}
,
\end{eqnarray*}
where we used Lemma~\ref{l-formabsch}\ref{l-formabsch-4}
and $c > 0$ is as in Proposition~\ref{pBS404}\ref{pBS404-4}.

Secondly and similarly
\[
\|(H \dotplus U + \lambda)^{-1} \, M_W \, (H \dotplus U + \lambda)^{-1}\|_\HS
\leq (2 \|(H + 1)^{-1/2}\|_{\cs_4})^2 \, c \, \|W\|_{L_2}
.  \]

Thirdly,
\[
\|g(H \dotplus V)\|_\HS
\leq \|g(H \dotplus V)\|_{\cs_1}^{1/2} \, \|g(H \dotplus V)\|_{L_2 \to L_2}^{1/2}
\leq c_2^{1/2} \, M_2^{1/2} 
,  \]
where $c_2$ and $M_2$ are as in Lemma~\ref{lBS402}.

Finally, using again Corollary~\ref{cBS405} we obtain
\begin{eqnarray*}
\lefteqn{
\|(H \dotplus U + \lambda)^{-1} \, M_W \, (H \dotplus U + \lambda)^{-1} 
         - (H \dotplus V + \lambda)^{-1} \, M_W \, (H \dotplus V + \lambda)^{-1}\|_\HS
} \hspace*{10mm}  \\*
& \leq & \|\Big( (H \dotplus U + \lambda)^{-1} - (H \dotplus V + \lambda)^{-1} \Big)
   M_W \, (H \dotplus U + \lambda)^{-1} \|_\HS  \\*
& & \hspace*{10mm} {} 
   + \|(H \dotplus V + \lambda)^{-1} \, M_W \Big( (H \dotplus U + \lambda)^{-1} - (H \dotplus V + \lambda)^{-1} \Big)\|_\HS  \\
& \leq & \|(H \dotplus U + \lambda)^{-1} \, M_{U-V} \, (H \dotplus V + \lambda)^{-1} \, 
   M_W \, (H \dotplus U + \lambda)^{-1}\|_\HS  \\*
& & \hspace*{10mm} {} 
   + \|(H \dotplus V + \lambda)^{-1} \, M_W \, (H \dotplus U + \lambda)^{-1} \, 
   M_{U-V} \, (H \dotplus V + \lambda)^{-1}\|_\HS
.
\end{eqnarray*}
Now 
\begin{eqnarray*}
\lefteqn{
\|(H \dotplus U + \lambda)^{-1} \, M_{U-V} \, (H \dotplus V + \lambda)^{-1} \, 
   M_W \, (H \dotplus U + \lambda)^{-1}\|_\HS
} \hspace*{30mm}  \\*
& \leq & \|(H \dotplus U + \lambda)^{-1}\|_\HS
   \, \|M_{U-V} \, (H \dotplus V + \lambda)^{-1}\|_{L_2 \to L_2} \cdot  \\*
& & \hspace*{10mm} {} \cdot
   \|M_W \, (H \dotplus U + \lambda)^{-1}\|_{L_2 \to L_2}  \\
& \leq & (2 \|(H + 1)^{-1/2}\|_{\cs_4})^2 \, c^2 \, \|U-V\|_{L_2} \, \|W\|_{L_2}
,
\end{eqnarray*}
where we used Lemma~\ref{l-formabsch}\ref{l-formabsch-4}
and $c > 0$ is as in Proposition~\ref{pBS404}.
Similarly 
\begin{eqnarray*}
\lefteqn{
\|(H \dotplus V + \lambda)^{-1} \, M_W \, (H \dotplus U + \lambda)^{-1} \, 
   M_{U-V} \, (H \dotplus V + \lambda)^{-1}\|_\HS
} \hspace*{50mm}  \\*
& \leq & (2 \|(H + 1)^{-1/2}\|_{\cs_4})^2 \, c^2 \, \|U-V\|_{L_2} \, \|W\|_{L_2}
.
\end{eqnarray*}

Together we obtain that there is a $C > 0$ such that 
\[
\Big| \int_\Omega \Big( \cm(U) - \cm(V) \Big) \, W \Big|
\leq C \, \|U-V\|_{L_2} \, \|W\|_{L_2}
\]
for all $W \in L_\infty(\Omega)$
and $U,V \in L_2(\Omega,\Ri)$ with $\|U\|_{L_2} \leq R$
and $\|V\|_{L_2} \leq R$.
Hence 
\[
\|\cm(U) - \cm(V)\|_{L_2} 
\leq C \, \|U-V\|_{L_2}
\]
for all $U,V \in L_2(\Omega,\Ri)$ with $\|U\|_{L_2} \leq R$
and $\|V\|_{L_2} \leq R$, as required.
\end{proof}

The map $\cm$ has the following monotonicity on 
the real Hilbert space $L_2(\Omega,\Ri)$.

\begin{prop} \label{pBS508}
Let $U,V \in L_2(\Omega,\Ri)$.
Then $(\cm(U) - \cm(V), U - V)_{L_2} \leq 0$.
\end{prop}
\begin{proof}
We follow arguments given in \cite{KaiserNeidhardtRehberg}.
First suppose that $U,V \in L_\infty(\Omega,\Ri)$.
There exists an orthonormal basis $(\psi_n)_{n \in \Ni}$ for 
$L_2(\Omega)$ and for all $n \in \Ni$ there is a $\lambda_n \in \Ri$ 
such that $(H \dotplus U) \psi_n = \lambda_n \, \psi_n$.
Similarly there exists an orthonormal basis $(\varphi_n)_{n \in \Ni}$ for 
$L_2(\Omega)$ and for all $n \in \Ni$ there is a $\mu_n \in \Ri$ 
such that $(H \dotplus V) \varphi_n = \mu_n \, \varphi_n$.

Let $n,m \in \Ni$.
Then $\psi_n \in D(H \dotplus U) = D(H) = D(H \dotplus V)$ by Proposition~\ref{pBS404}\ref{pBS404-1}.
Hence 
\[
((U - V) \psi_n, \varphi_m)_{L_2}
= ((H \dotplus U - H \dotplus V) \psi_n, \varphi_m)_{L_2}
= (\lambda_n - \mu_m) \, (\psi_n, \varphi_m)_{L_2}
.  \]
Then 
\begin{eqnarray*}
\int_\Omega (U - V) \, \cm(U)
& = & \tr( M_{U-V} \, f(H \dotplus U)
= \sum_{n = 1}^\infty (M_{U-V} \, f(H \dotplus U) \psi_n, \psi_n)_{L_2}  \\
& = & \sum_{n = 1}^\infty f(\lambda_n) \, (( U - V) \psi_n, \psi_n)_{L_2}  \\
& = & \sum_{n,m = 1}^\infty f(\lambda_n) \, (( U - V) \psi_n, \varphi_m)_{L_2} (\varphi_m, \psi_n)_{L_2}  \\
& = & \sum_{n,m = 1}^\infty f(\lambda_n) \, (\lambda_n - \mu_m) \, |(\psi_n, \varphi_m)_{L_2}|^2
.
\end{eqnarray*}
Similarly
\[
\int_\Omega (V - U) \, \cm(V)
= \sum_{n,m = 1}^\infty f(\mu_m) \, (\mu_m - \lambda_n) \, |(\psi_n, \varphi_m)_{L_2}|^2
.  \]
Hence 
\[
\int_\Omega (U - V) \, \Big( \cm(U) - \cm(V) \Big)
= \sum_{n,m = 1}^\infty (f(\lambda_n) - f(\mu_m)) \, (\lambda_n - \mu_m) \, |(\psi_n, \varphi_m)_{L_2}|^2
\leq 0
\]
since $f$ is decreasing.

Now the proposition follows from Proposition~\ref{p-lipschpartdens}\ref{p-lipschpartdens-2} 
by density and continuity.
\end{proof}

If $V \in L_2(\Omega,\Ri)$ and $t \in \Ri$, 
then $H \dotplus V - t = H + (V - t \, \one_\Omega)$ and 
$V - t \, \one_\Omega \in L_2(\Omega,\Ri)$.
Choosing $k = 0$ in Lemma~\ref{lBS402}\ref{lBS402-2} 
gives that $f(H \dotplus V - t)$
is trace class for all $V \in L_2(\Omega,\Ri)$ and $t \in \Ri$.

\begin{lemma} \label{l-Fermi}
For all $V \in L_2(\Omega,\Ri)$ and $\widetilde N \in (0,\infty)$
there is a unique
$t \in \Ri$ such that $\tr \bigl ( f(H \dotplus V - t) \bigr ) = \widetilde N$.
\end{lemma}
\begin{proof}
Let $\lambda_1 \leq \lambda_2 \leq \ldots$ be the eigenvalues of the 
operator $H \dotplus V$, repeated with multiplicity. 
Then 
\[
 \tr \bigl ( f(H \dotplus V - t) \bigr ) 
= \sum_{n=1}^\infty f(\lambda_n - t)
\]
for all $t \in \Ri$.
This series is absolutely convergent for each $t \in \Ri$.
Since $f$ is continuous and strictly decreasing, it follows from 
the Lebesgue dominated convergence theorem that $t \mapsto \tr \bigl ( f(H \dotplus V - t) \bigr )$
is continuous and strictly increasing.

Because $f(0) > 0$ one deduces that 
$\lim_{t \to \infty} \tr \bigl ( f(H \dotplus V - t) \bigr ) = \infty$.
Using again the Lebesgue dominated convergence theorem
it follows that 
$\lim_{t \to - \infty} \tr \bigl ( f(H \dotplus V - t) \bigr ) = 0$.
Then the existence follows.
\end{proof}

For all $V \in L_2(\Omega,\Ri)$ we define the {\bf Fermi level}
$\mathcal E(V) = t$, where $t \in \Ri$ is such that 
\[
\tr \bigl ( f(H \dotplus V - t) \bigr ) = N
.  \]
Recall that $N$ is fixed in Theorem~\ref{tBS510}.
We next prove that the function $V \mapsto \ce(V)$ is locally bounded.

\begin{prop} \label{pBS506}
For all $R > 0$ there exists an $M > 0$ such that 
$|\ce(V)| \leq M$ for all $V \in L_2(\Omega,\Ri)$ with $\|V\|_{L_2} \leq R$.
\end{prop}
\begin{proof}
The proof is a modification of the proof of Lemma~\ref{l-Fermi}.
Write $\lambda = 1 + \gamma \, R^4$.
For all $V \in L_2(\Omega,\Ri)$ with $\|V\|_{L_2} \leq R$ let 
$\lambda_1^{(V)} \leq \lambda_2^{(V)} \leq \ldots$ be the eigenvalues of the 
operator $H \dotplus V$, repeated with multiplicity. 
Then the mini-max theorem together with the form bounds of Lemma~\ref{l-formabsch}\ref{l-formabsch-5}
gives
\[
\tfrac{1}{4} \, \lambda_n^{(0)} - \lambda 
\leq \lambda_n^{(V)}
\leq \tfrac{7}{4} \, \lambda_n^{(0)} + \lambda 
\]
for all $n \in \Ni$.
Let $t \in \Ri$.
Since $f$ is decreasing one estimates
\[
\sum_{n=1}^\infty f(\tfrac{7}{4} \, \lambda_n^{(0)} + \lambda - t)
\leq \sum_{n=1}^\infty f(\lambda_n^{(V)} - t)
\leq \sum_{n=1}^\infty f(\tfrac{1}{4} \, \lambda_n^{(0)} - \lambda - t)
.  \]
Arguing as in the proof of Lemma~\ref{l-Fermi} there are $T,\widetilde T \in \Ri$
such that 
\[
\sum_{n=1}^\infty f(\tfrac{7}{4} \, \lambda_n^{(0)} + \lambda - T)
= N
= \sum_{n=1}^\infty f(\tfrac{1}{4} \, \lambda_n^{(0)} - \lambda - \widetilde T)
.  \]
Then $T \geq \ce(V) \geq \widetilde T$.
\end{proof}

Define the {\bf particle density} $\cn \colon L_2(\Omega,\Ri) \to L_2(\Omega,\Ri)$
by 
\[
\mathcal N(V)
= \cm(V - \ce(V) \, \one_\Omega)
.  \]
So 
\[
\int_\Omega \cn(V) \, W 
= \tr \bigl ( M_W \, f(H \dotplus V - \ce(V)) \bigr )
\]
and 
\begin{equation}
\int_\Omega \cn(V)
= \tr f(H \dotplus V - \ce(V)) 
= N
\label{etBS506;1}
\end{equation}
for all $V \in L_2(\Omega,\Ri)$ and $W \in L_\infty(\Omega)$.
We next prove that Proposition~\ref{p-lipschpartdens}\ref{p-lipschpartdens-2} 
remains valid if 
$\mathcal M$ is replaced by $\mathcal N$.
This was proved before in \cite{KaiserRehberg} Section~4.

\begin{prop} \label{pBS506-N}
Let $R > 0$. 
Then there exists an $M > 0$ such that 
\[
\|\cn(U) - \cn(V) \|_{L_2} \leq M \, \|U - V \|_{L_2}
\]
for all $U,V \in L_2(\Omega,\Ri)$ with $\|U\|_{L_2} \leq R$
and $\|V\|_{L_2} \leq R$.
\end{prop}
\begin{proof}
Using Propositions~\ref{pBS506} and \ref{p-lipschpartdens}\ref{p-lipschpartdens-2}
it suffices to show that there is an $M > 0$ such that 
$|\ce(U) - \ce(V)| \leq M \, \|U - V\|_{L_2}$
for all $U,V \in L_2(\Omega,\Ri)$ with $\|U\|_{L_2} \leq R$ and $\|V\|_{L_2} \leq R$.

For all $V \in L_2(\Omega,\Ri)$ let 
$\lambda_1^{(V)} \leq \lambda_2^{(V)} \leq \ldots$ be the eigenvalues of the 
operator $H \dotplus V$, repeated with multiplicity. 
Then the mini-max theorem together with the form bounds of Lemma~\ref{l-formabsch}\ref{l-formabsch-5}
gives
\[
\tfrac{1}{4} \, \lambda_1^{(0)} - \lambda 
\leq \lambda^{(V)}_1
\leq \tfrac{7}{4} \, \lambda_1^{(0)} + \lambda 
\]
for all $V \in L_2(\Omega,\Ri)$ with $\|V\|_{L_2} \leq R$,
where $\lambda = 1 + \gamma \, R^4$.
By Proposition~\ref{pBS506} there is an $M_1 > 0$ such that 
$|\ce(V)| \leq M_1$ for all $V \in L_2(\Omega,\Ri)$ with $\|V\|_{L_2} \leq R$.
Since $f$ is locally bounded, there is an $M_2 > 0$ such that 
$f(\lambda_1^{(V)} - \ce(V)) \leq M_2$ for all $V \in L_2(\Omega,\Ri)$ with $\|V\|_{L_2} \leq R$.
By the assumptions on~$f$ 
there is an $L > 0$ such that
$|t-s| \leq L \, |f(t) - f(s)|$ for all $t,s \in \Ri$ with 
$f(t) \in [-M_2,M_2]$ and $f(s) \in [-M_2,M_2]$.
By Proposition~\ref{p-lipschpartdens}\ref{p-lipschpartdens-2} there is a $c > 0$ 
such that 
\[
\|\cm(U) - \cm(V)\|_{L_2}
\leq c \, \|U - V\|_{L_2}
\]
for all $U,V \in L_2(\Omega,\Ri)$ with $\|U\|_{L_2} \leq R + M_1 \, |\Omega|^{1/2}$
and $\|V\|_{L_2} \leq R + M_1 \, |\Omega|^{1/2}$.

Now let $U,V \in L_2(\Omega,\Ri)$ with $\|U\|_{L_2} \leq R$ and 
$\|V\|_{L_2} \leq R$.
Without loss of generality we may assume that $\ce(V) \geq \ce(U)$.
Since $f$ is decreasing one deduces that 
\[
f(\lambda_n^{(V)} - \ce(V)) - f(\lambda_n^{(V)} - \ce(U))
\geq 0
\]
for all $n \in \Ni$.
Consequently 
\begin{eqnarray*}
\lefteqn{
\Tr f(H \dotplus V - \ce(V)) - \Tr f(H \dotplus V - \ce(U))
} \hspace*{30mm} \\*
& = & \sum_{n=1}^\infty \Big( f(\lambda_n^{(V)} - \ce(V)) - f(\lambda_n^{(V)} - \ce(U)) \Big)  \\
& \geq & |f(\lambda_1^{(V)} - \ce(V)) - f(\lambda_1^{(V)} - \ce(U))|  \\
& \geq & L^{-1} \, |(\lambda_1^{(V)} - \ce(V)) - (\lambda_1^{(V)} - \ce(U))|  \\
& = & L^{-1} \, |\ce(U) - \ce(V)|
.
\end{eqnarray*}
On the other hand,
\[
\Tr f(H \dotplus V - \ce(V))
= N
= \Tr f(H \dotplus U - \ce(U))
\]
and hence
\begin{eqnarray*}
\lefteqn{
\Tr f(H \dotplus V - \ce(V)) - \Tr f(H \dotplus V - \ce(U))
} \hspace*{30mm} \\*
& = & |\Tr f(H \dotplus U - \ce(U)) - \Tr f(H \dotplus V - \ce(U))|  \\
& = & \Big| \int_\Omega \Big( \cm(U - \ce(U)) - \cm(V - \ce(U)) \Big) \Big|  \\
& \leq & |\Omega|^{1/2} \, \|\cm(U - \ce(U)) - \cm(V - \ce(U))\|_{L_2}  \\
& \leq & c \, |\Omega|^{1/2} \, \|U - V\|_{L_2}
\end{eqnarray*}
from which the theorem follows.
\end{proof}

Also Proposition~\ref{pBS508} is valid for $\cn$ instead of $\cm$.

\begin{prop} \label{pBS509}
Let $U,V \in L_2(\Omega,\Ri)$.
Then $(\cn(U) - \cn(V), U - V)_{L_2} \leq 0$.
\end{prop}
\begin{proof}
Using (\ref{etBS506;1}) and Proposition~\ref{pBS508} one deduces that 
\begin{eqnarray*}
\lefteqn{
(\cn(U) - \cn(V), U - V)_{L_2}
} \hspace*{10mm} \\*
& = & (\cn(U) - \cn(V), (U - \ce(U) \, \one_\Omega) - (V - \ce(V) \, \one_\Omega) )_{L_2}  \\
& = & ((\cm(U - \ce(U) \, \one_\Omega) - \cm(V - \ce(V) \, \one_\Omega), 
            (U - \ce(U) \, \one_\Omega) - (V - \ce(V) \, \one_\Omega) )_{L_2}
\leq 0
\end{eqnarray*}
as required.
\end{proof}

We next show that for all $V_0 \in L_2(\Omega,\Ri)$ and $q \in (W^{1,2}_D(\Omega,\Ri))^*$
there is a $V \in W^{1,2}_D(\Omega)$ such that 
\[
- \nabla \cdot \varepsilon \nabla V
= q + \cn(V_0 + V)
\]
in a weak sense.

\begin{prop} \label{tBS509.5}
Let $V_0 \in L_2(\Omega,\Ri)$.
Write $\ch = W^{1,2}_D(\Omega,\Ri)$.
Define the map $A_0 \colon \ch \to \ch^*$ by
\[
\langle A_0 V, W \rangle_{\ch^* \times \ch}
= \int_\Omega \varepsilon \nabla V \cdot \nabla W
   - (\cn(V_0 + V), W)_{L_2(\Omega)}
.  \]
Then one has the following.
\begin{tabel}
\item \label{tBS509.5-1}
The operator $A_0$ is strongly monotone with monotonicity constant
$\mu \, (1 + c_P)^{-1}$.
\item \label{tBS509.5-2}
For all $R > 0$ there exists a $C > 0$ such that 
\[
\|A_0 U - A_0 V\|_{\ch^*}
\leq C \, \|U - V\|_\ch
\]
for all $U,V \in \ch$ with $\|U\|_\ch \leq R$ and $\|V\|_\ch \leq R$.
\item \label{tBS509.5-3}
For all $q \in \ch^*$ there is a unique $V \in \ch$ 
such that 
\[
\langle A_0 V, W \rangle_{\ch^* \times \ch}
= \langle q, W \rangle_{\ch^* \times \ch}
\]
for all $W \in \ch$.
It follows that 
$\|V\|_\ch \leq \frac{1+c_P}{\mu} \, \|q + \cn(V_0)\|_{\ch^*}$.
\end{tabel}
\end{prop}
\begin{proof}
`\ref{tBS509.5-1}'.
Let $U,V \in \ch$.
Then 
\begin{eqnarray*}
\lefteqn{
\langle A_0 U - A_0 V, U - V \rangle_{\ch^* \times \ch}
} \hspace*{10mm}  \\
& \geq & \mu \int_\Omega |\nabla (U - V)|^2
   - (\cn(V_0 + U) - \cn(V_0 + V), (V_0 + U) - (V_0 + V))_{L_2}  \\
& \geq & \mu \, (1 + c_P)^{-1} \|U - V\|_\ch^2
,  
\end{eqnarray*}
where we used Proposition~\ref{pBS509}.

`\ref{tBS509.5-2}'.
This follows from Proposition~\ref{pBS506-N} and the boundedness 
of the coefficient function~$\varepsilon$.

`\ref{tBS509.5-3}'.
This follows immediately from the previous statements and Theorem~\ref{t-ggz}.
\end{proof}

Write $\Psi(V_0,q) = V$ if $V_0$, $q$ and $V$ are as in Proposition~\ref{tBS509.5}\ref{tBS509.5-3}.

\begin{proof}[{\bf Proof of Theorem~\ref{tBS510}\ref{tBS510-8}.}]
If follows from Theorem~\ref{t-ggz}\ref{t-ggz-4} that 
\[
\|\Psi(V_0,q) - \Psi(V_0,\tilde q)\|_{W^{1,2}_D(\Omega)}
\leq \frac{1 + c_P}{\mu} \, \|q - \tilde q\|_{(W^{1,2}_D(\Omega))^*}
.  \]
Hence it remains to estimate $\|\Psi(V_0,q) - \Psi(V_1,q)\|_{W^{1,2}_D(\Omega)}$.
Write $\ch = W^{1,2}_D(\Omega,\Ri)$.
Define the maps $A_0, A_1 \colon \ch \to \ch^*$ by
\begin{eqnarray*}
\langle A_0 V, W \rangle_{\ch^* \times \ch}
& = & \int_\Omega \varepsilon \nabla V \cdot \nabla W
   - (\cn(V_0 + V), W)_{L_2(\Omega)}
\quad \mbox{and}  \\
\langle A_1 V, W \rangle_{\ch^* \times \ch}
& = & \int_\Omega \varepsilon \nabla V \cdot \nabla W
   - (\cn(V_1 + V), W)_{L_2(\Omega)}
.  
\end{eqnarray*}
Write $m = \frac{\mu}{1 + c_P}$.
Let 
\[
R = \frac{2}{m} \max \{ \|q + \cn(V_0)\|_{\ch^*}, \|q + \cn(V_1)\|_{\ch^*} \}
.  \]
By Proposition~\ref{tBS509.5}\ref{tBS509.5-2} there exists an $M > 0$ such that 
\[
\|A_0 U - A_0 V\|_{\ch^*} \leq M \, \|U - V\|_\ch
\quad \mbox{and} \quad
\|A_1 U - A_1 V\|_{\ch^*} \leq M \, \|U - V\|_\ch
\]
for all $U,V \in \ch$ with $\|U\|_\ch \leq R$ and $\|V\|_\ch \leq R$.
By Proposition~\ref{pBS506-N} there exists a $c > 0$ such that 
\[
\|\cn(U) - \cn(V)\|_{L_2} 
\leq c \, \|U - V\|_{L_2}
\]
for all $U,V \in L_2(\Omega,\Ri)$ such that 
$\|U\|_{L_2} \leq R + \|V_0\|_{L_2} + \|V_1\|_{L_2}$ and 
$\|V\|_{L_2} \leq R + \|V_0\|_{L_2} + \|V_1\|_{L_2}$.
Let $J \colon \ch \to \ch^*$ be the duality map.
Define $Q_0,Q_1 \colon \{ U \in \ch : \|U\|_\ch \leq R \} \to \ch$ by
\begin{equation}
Q_0 U = U - \frac {m}{M^2} \, J^{-1}(A_0 U - q)
\quad \mbox{and} \quad
Q_1 U = U - \frac {m}{M^2} \, J^{-1}(A_1 U - q)
.  
\label{etBS510;20}
\end{equation}

Now we estimate $\|\Psi(V_0,q) - \Psi(V_1,q)\|_{W^{1,2}_D(\Omega)}$.
Write $U = \Psi(V_0,q)$ and $V = \Psi(V_1,q)$.
Then $\|U\|_{L_2} \leq R$ and $\|V\|_{L_2} \leq R$ by Proposition~\ref{tBS509.5}\ref{tBS509.5-3}.
Moreover,
$Q_0 U = U$ and $Q_1 V = V$.
Using Proposition~\ref{tBS509.5}\ref{tBS509.5-1}
and Theorem~\ref{t-ggz}\ref{t-ggz-1} one estimates
\begin{eqnarray*}
\|U - V\|_\ch
& = & \|Q_0 U - Q_1 V\|_\ch
= \|Q_0 U - Q_0 V\|_\ch + \|Q_0 V - Q_1 V\|_\ch  \\
& \leq & \sqrt{ 1 - \frac {m^2}{M^2}} \, \|U - V\|_\ch 
   + \frac{m}{M^2} \, \| J^{-1} ( \cn(V_0 + V) - \cn(V_1 + V) ) \|_{\ch}  
.
\end{eqnarray*}
Rearrangement gives
\begin{eqnarray*}
\|U - V\|_\ch
& \leq & \Big( 1 - \sqrt{ 1 - \frac {m^2}{M^2}} \Big)^{-1} \, 
   \frac{m}{M^2} \, \| \cn(V_0 + V) - \cn(V_1 + V) \|_{\ch^*}   \\
& \leq & \Big( 1 - \sqrt{ 1 - \frac {m^2}{M^2}} \Big)^{-1} \, 
   \frac{m}{M^2} \, \|I\|_{L_2 \to \ch^*} \, \|\cn(V_0 + V) - \cn(V_1 + V)\|_{L_2}  \\
& \leq & \Big( 1 - \sqrt{ 1 - \frac {m^2}{M^2}} \Big)^{-1} \, 
   \frac{m}{M^2} \, \|I\|_{L_2 \to \ch^*} \, c \, \|V_0 - V_1\|_{L_2}  
\end{eqnarray*}
as required.
\end{proof}

Finally we consider regularity of the solutions
under some weak additional assumptions. 
For all $p \in (1,\infty)$ let
$W^{1,p}_D(\Omega)$ be the $W^{1,p}(\Omega)$-closure of the set 
\[
\{ u|_\Omega : u \in C^\infty(\R^d) \mbox{ and } \supp u \cap D = \emptyset \} 
.  \]
Let $p'$ be the dual exponent of $p$.
Define $W^{-1,p'}_D(\Omega)$ to be the dual space of $W^{1,p}_D(\Omega)$.

\begin{theorem} \label{tBS524}
Let $V_0 \in L_2(\Omega,\Ri)$, $p \in [2,\infty]$, and $q \in W^{-1,p}_D(\Omega,\Ri)$.
\mbox{}
\begin{tabel}
\item \label{tBS523-1}
If $p > d$, then $\Psi(V_0,q) \in L_\infty(\Omega)$.
\item \label{tBS523-2}
If $2 < p \leq d = 3$, then $\Psi(V_0,q) \in L_{p^*}(\Omega)$,
where $\frac{1}{p^*} = \frac{1}{p} - \frac{1}{3}$.
\item \label{tBS525-2}
Suppose $p > d$.
Under a measure theoretic condition of the relative boundary of $D$ in $\partial \Omega$
(see \cite{ERe2} Assumption~(III) in Theorem~1.1 )
it follows that $\Psi(V_0,q)$ is H\"older continuous
for all $V_0 \in L_2(\Omega,\Ri)$ and $q \in W^{-1,p}_D(\Omega)$.
\item \label{tBS525-3}
Suppose that $D = \partial \Omega$ and that the 
coefficient function $\varepsilon$ is Lipschitz continuous.
Further suppose that $\Omega$ is of class $C^{1,1}$ or convex.
Then $\Psi(V_0,q) \in W^{2,2}(\Omega)$ if $q \in L_2(\Omega,\Ri)$.
\end{tabel}
\end{theorem}
\begin{proof}
The Sobolev embedding theorem gives $L_2(\Omega) \subset W^{-1,6}_D(\Omega)$.
Write $\tilde p = \min(p,6)$ and $V = \Psi(V_0,q)$.
Then $q + \cn(V_0 + V) \in W^{-1,\tilde p}_D(\Omega)$.
Moreover, 
\begin{equation}
\int_\Omega \varepsilon \nabla V \cdot \nabla W
= \langle q + \cn(V_0 + V), W \rangle_{\cd'(\Omega) \times \cd(\Omega)}
\label{etBS524;1}
\end{equation}
for all $W \in C_c^\infty(\Omega)$.

`\ref{tBS523-1}'.
Now $\tilde p > d$ and \cite{Stam2} Th\'eor\`eme~4.2(a) implies that $V \in L_\infty(\Omega)$.

`\ref{tBS523-2}'.
Now $\tilde p = p$ and one can use \cite{Stam2} Th\'eor\`eme~4.2(b).

`\ref{tBS525-2}'.
See \cite{ERe2} Theorem~1.1.

`\ref{tBS525-3}'.
We use (\ref{etBS524;1}).
Then the regularity follows from \cite{Gris} Theorem~2.2.2.3
in case $\Omega$ is of class $C^{1,1}$ and 
from \cite{Gris} Theorem~3.2.1.2 in case $\Omega$ is convex.
\end{proof}

\begin{theorem} \label{tBS526}
Suppose that the set $D$ is a $(d-1)$-set in the sense of Jonsson--Wallin
(\cite{JW} Chapter~II).
Then there exists a $p > 2$ such that $\Psi(V_0,q) \in W^{1,p}_D(\Omega)$
for all $V_0 \in L_2(\Omega,\Ri)$ and $q \in W^{-1,p}_D(\Omega)$.
\end{theorem}
\begin{proof}
This follows from \cite{HJKR} Theorem~5.6.
\end{proof}

If $V_0 \in L_2(\Omega,\Ri)$ and $q \in (W^{1,2}_D(\Omega,\Ri))^*$,
then the solution $\Psi(V_0,q)$ of the Schr\"o\-din\-ger--Poisson system obtained
in Theorem~\ref{tBS510}\ref{tBS510-5}, or Proposition~\ref{tBS509.5}\ref{tBS509.5-3},
is constructed  via a contraction in Theorem~\ref{t-ggz}\ref{t-ggz-3}.
Explicitly, if $Q_0$ is as in (\ref{etBS510;20}), $U_1 = 0$ and 
inductively $U_{n+1} = Q_0 U_n$ for all $n \in \Ni$, then 
$\Psi(V_0,q) = \lim_{n \to \infty} U_n$ in $W^{1,2}_D(\Omega)$.
We do not know whether the convergence is valid in any of the 
spaces mentioned in Theorems~\ref{tBS524} and~\ref{tBS526}.

\subsection*{Acknowledgment} The third named author (J.R.) 
wants to thank Alexander Strohmaier (Hannover) 
for valuable discussions on the subject. 
Also we wish to thank Angel Paredes and Humberto Michinel (Universidade de Vigo).

\small

\noindent 
\emph{V. Bach}, Institut f\"ur Analysis und Algebra,
Carl-Friedrich-Gauss-Fakult\"at,
Technische Universit\"at Braunschweig,
Universit\"tsplatz 2,
38106 Braunschweig,
Germany,  \\
\texttt{v.bach@tu-braunschweig.de}

\medskip 

\noindent 
\emph{A.F.M. ter Elst}, Department of Mathematics, University of Auckland, 
Private Bag 92019, Auckland 1142, New Zealand,  \\
\texttt{terelst@math.auckland.ac.nz}

\medskip 

\noindent
\emph{J. Rehberg},
Weierstrass Institute for Applied Analysis and Stochastics,
Mohrenstr.~39, 
D-10117 Berlin, 
Germany,  \\
\texttt{rehberg@wias-berlin.de}


\begin{thebibliography}{HJKR12}

\bibitem[AE]{AE1}
{\sc Arendt, W. {\rm and} Elst, A. F.~M. ter}, Gaussian estimates for second
  order elliptic operators with boundary conditions.
\newblock {\em J. Operator Theory} {\bf 38} (1997),  87--130.

\bibitem[BS1]{BirmanSolomyak4}
{\sc Birman, M.~S. {\rm and} Solomyak, M.}, Stieltjes double operator
  integrals.
\newblock {\em Sov. Math., Dokl.} {\bf 6} (1965),  1567--1571.
\newblock Translation from Dokl. Akad. Nauk SSSR 165, 1223-1226 (1965).

\bibitem[BS2]{BirmanSolomyak2}
\leavevmode\vrule height 2pt depth -1.6pt width 23pt, Double Stieltjes operator
  integrals.
\newblock {\em Probl. Mat. Fiz.} {\bf 1} (1966),  33--67.

\bibitem[BS3]{BirmanSolomyak3}
\leavevmode\vrule height 2pt depth -1.6pt width 23pt, Stieltjes double operator
  integrals and multiplier problems.
\newblock {\em Sov. Math., Dokl.} {\bf 7} (1966),  1618--1621.
\newblock Translation from Dokl. Akad. Nauk SSSR 171, 1251-1254 (1966).

\bibitem[BS4]{BirmanSolomyak}
\leavevmode\vrule height 2pt depth -1.6pt width 23pt, Double operator integrals
  in a Hilbert space.
\newblock {\em Integral Equations Operator Theory} {\bf 47} (2003),  131--168.

\bibitem[ET]{EgertTolksdorf}
{\sc Egert, M. {\rm and} Tolksdorf, P.}, Characterizations of Sobolev functions
  that vanish on a part of the boundary.
\newblock {\em Discrete Contin. Dyn. Syst. Ser. S} {\bf 10} (2017),  729--743.

\bibitem[ER]{ERe2}
{\sc Elst, A. F.~M. ter {\rm and} Rehberg, J.}, H{\"o}lder estimates for
  second-order operators on domains with rough boundary.
\newblock {\em Adv. Diff. Equ.} {\bf 20} (2015),  299--360.

\bibitem[GGZ]{GGZ}
{\sc Gajewski, H., Gr{\"o}ger, K. {\rm and} Zacharias, K.}, {\em Nichtlineare
  Operatorgleichungen und Operatordifferentialgleichungen}.
\newblock Mathematische Lehr\-b{\"u}cher und Monographien, II. Abteilung
  Mathematische Monographien 38. Akademie-Verlag, Berlin, 1974.

\bibitem[GN]{GesztesyNichols}
{\sc Gesztesy, F. {\rm and} Nichols, R.}, Some applications of almost analytic
  extensions to operator bounds in trace ideals.
\newblock {\em Methods Funct. Anal. Topology} {\bf 21} (2015),  151--169.

\bibitem[GPS]{GesztesyPushnitskiSimon}
{\sc Gesztesy, F., Pushnitski, A. {\rm and} Simon, B.}, On the Koplienko
  spectral shift function. I. Basics.
\newblock {\em J. Math. Phys Anal. Geom.} {\bf 4} (2008),  67--107.

\bibitem[Gri]{Gris}
{\sc Grisvard, P.}, {\em Elliptic problems in nonsmooth domains}.
\newblock Monographs and Studies in Mathematics 24. Pitman, Boston etc., 1985.

\bibitem[HJKR]{HJKR}
{\sc Haller-Dintelmann, R., Jonsson, A., Knees, D. {\rm and} Rehberg, J.},
  Elliptic and parabolic regularity for second-order divergence operators with
  mixed boundary conditions.
\newblock {\em Math. Methods Appl. Sci.} {\bf 39} (2016),  5007--5026.

\bibitem[JW]{JW}
{\sc Jonsson, A. {\rm and} Wallin, H.}, Function spaces on subsets of ${\bf
  R}^n$.
\newblock {\em Math. Rep.} {\bf 2}, No.\ 1 (1984).

\bibitem[Kat]{Kat1}
{\sc Kato, T.}, {\em Perturbation theory for linear operators}.
\newblock Second edition, Grund\-lehren der mathematischen Wissenschaften 132.
  Springer-Verlag, Berlin etc., 1980.

\bibitem[KNR]{KaiserNeidhardtRehberg}
{\sc Kaiser, H.-C., Neidhardt, H. {\rm and} Rehberg, J.}, Monotonicity
  properties of the quantum mechanical particle density: an elementary proof.
\newblock {\em Monatsh. Math.} {\bf 158} (2009),  179--185.

\bibitem[KR1]{KaiserRehberg2}
{\sc Kaiser, H.-C. {\rm and} Rehberg, J.}, About a one-dimensional stationary
  Schr{\"o}dinger--Poisson system with Kohn--Sham potential.
\newblock {\em Z. angew. Math. Phys.} {\bf 50} (1999),  423--458.

\bibitem[KR2]{KaiserRehberg}
\leavevmode\vrule height 2pt depth -1.6pt width 23pt, About a stationary
  Schr{\"o}dinger--Poisson system with Kohn--Sham potential in a bounded two-
  or three-dimensional domain.
\newblock {\em Nonlinear Anal.} {\bf 41} (2000),  33--72.

\bibitem[L{\"o}w]{Lowner}
{\sc L{\"o}wner, K.}, {\"U}ber monotone Matrixfunktionen.
\newblock {\em Math. Z.} {\bf 38} (1934),  177--216.

\bibitem[Nie]{Nier}
{\sc Nier, F.}, A variational formulation of Schr{\"o}dinger--Poisson systems
  in dimension $d \leq 3$.
\newblock {\em Comm. Partial Differential Equations} {\bf 18} (1993),
  1125--1147.

\bibitem[POM]{ParedesOlivieriMichinel}
{\sc Paredes, A., Olivieri, D.~N. {\rm and} Michinel, H.}, From optics to dark
  matter: A review on nonlinear Schr{\"o}dinger--Poisson systems.
\newblock {\em Physica D} {\bf 403} (2020),  132301.

\bibitem[Pel]{Peller}
{\sc Peller, V.~V.}, An elementary approach to operator Lipschitz type
  estimates.
\newblock In {\em 50 years with Hardy spaces}, Operator Theory: Advances and
  Applications 261. Birkh{\"a}user, Cham, 2018.

\bibitem[RS]{RS1}
{\sc Reed, M. {\rm and} Simon, B.}, {\em Methods of modern mathematical physics
  I. Functional analysis}.
\newblock Academic Press, New York etc., 1972.

\bibitem[Sta]{Stam2}
{\sc Stampacchia, G.}, Le probl\`eme de Dirichlet pour les \'equations
  elliptiques du second ordre \`a coefficients discontinus.
\newblock {\em Ann. Inst. Fourier, Grenoble} {\bf 15} (1965),  189--258.

\bibitem[WPH]{WoodsPayneHasnip}
{\sc Woods, N.~D., Payne, M.~C. {\rm and} Hasnip, P.~J.}, Computing the
  self-consistent field in Kohn--Sham density functional theory.
\newblock {\em J. Phys.: Condens. Matter} {\bf 31} (2019),  453001.

\bibitem[Zei]{ZeidlerIIB}
{\sc Zeidler, E.}, {\em Nonlinear functional analysis and its applications.
  II/B. Nonlinear monotone operators}.
\newblock Springer-Verlag, New York, 1990.

\end{thebibliography}
\end{document}